\documentclass[10pt,journal,final,finalsubmission,twocolumn]{IEEEtran}
\usepackage{graphicx}
\usepackage{amssymb}
\usepackage{multirow}
\usepackage{algorithm}
\usepackage{algorithmic}
\usepackage{amsmath,amsfonts}

\newcommand{\note}[1]{}
\newtheorem{theorem}{Theorem}[section]
\newtheorem{lemma}{Lemma}[section]

\newtheorem{corollary}{Corollary}[section]
\newtheorem{problem}{Problem}[section]

\newcommand{\Id}{\mathtt{I}}

\newcommand{\trace}[1]{\ensuremath{\mathrm{tr}\left(#1\right)}}
\newcommand{\diag}[1]{\ensuremath{\mathrm{diag}\left(#1\right)}}
\newcommand{\spn}[1]{\ensuremath{\mathrm{span}\left(#1\right)}}
\newcommand{\rank}[1]{\ensuremath{\mathrm{rank}\left(#1\right)}}
\newcommand{\norm}[1]{\ensuremath{\left\| #1 \right\|}}

\newcommand{\qed}{\nobreak \ifvmode \relax \else
      \ifdim\lastskip<1.5em \hskip-\lastskip
      \hskip1.5em plus0em minus0.5em \fi \nobreak
      \vrule height0.5em width0.5em depth0.1em\fi}
\begin{document}
\title{Robust Recovery of Subspace Structures by Low-Rank Representation}
\author{Guangcan Liu$^{\dag,^\ddag,\sharp}$, Member, IEEE, Zhouchen Lin$^{\flat}$\footnote{Corresponding Author}, Senior Member, IEEE, Shuicheng Yan$^\ddag$, Senior Member, IEEE, Ju Sun$^\natural$, Student Member, IEEE, Yong Yu$^\dag$, and Yi Ma$^{\S,\sharp}$, Senior Member, IEEE\\
$^\dag$ Computer Science and Engineering, Shanghai Jiao Tong University, Shanghai, China\\
$^\S$ Visual Computing Group, Microsoft Research Asia, Beijing, China\\
$^\flat$ Key Laboratory of Machine Perception, Peking University, Beijing, China\\
$^\ddag$ Electrical and Computer Engineering, National University of Singapore, Singapore\\
$^\sharp$ Coordinated Science Laboratory, University of Illinois at Urbana-Champaign, USA\\
$^\natural$ Electrical Engineering, Columbia University, USA}
\maketitle

\begin{abstract}In this work we address the subspace clustering problem. Given a set of data samples (vectors) approximately drawn from a union of multiple subspaces, our goal is to cluster the samples into their respective subspaces and remove possible outliers as well. To this end, we propose a novel objective function named Low-Rank Representation (LRR), which seeks the lowest-rank representation among all the candidates that can represent the data samples as linear combinations of the bases in a given dictionary. It is shown that the convex program associated with LRR solves the subspace clustering problem in the following sense: when the data is clean, we prove that LRR exactly recovers the true subspace structures; when the data are contaminated by outliers, we prove that under certain conditions LRR can exactly recover the row space of the original data and detect the outlier as well; for data corrupted by arbitrary sparse errors, LRR can also approximately recover the row space with theoretical guarantees. Since the subspace membership is provably determined by the row space, these further imply that LRR can perform robust subspace clustering and error correction, in an efficient and effective way.
\end{abstract}

\section{Introduction}\label{sec:intro}
In pattern analysis and signal processing, an underlying tenet is that the data often contains some type of \emph{structure} that enables intelligent representation and processing. So one usually needs a parametric model to characterize a given set of data. To this end, the well-known (linear) \emph{subspaces} are possibly the most common choice, mainly because they are easy to compute and often effective in real applications. Several types of visual data, such as motion \cite{ijcv_1998_multibody,eccv_2006_lsa,motion_pami_2010_Rene}, face \cite{liu:2011:iccv} and texture \cite{tpami_2008_acl}, have been known to be well characterized by subspaces. Moreover, by applying the concept of reproducing kernel Hilbert space, one can easily extend the linear models to handle nonlinear data. So the subspace methods have been gaining much attention in recent years. For example, the widely used Principal Component Analysis (PCA) method and the recently established matrix completion \cite{CandesPIEEE} and recovery \cite{journal_2009_rpca2} methods are essentially based on the hypothesis that the data is approximately drawn from a low-rank subspace. However, a given data set can seldom be well described by a \emph{single} subspace. A more reasonable model is to consider data as lying near \emph{several} subspaces, namely the data is considered as samples approximately drawn from a mixture of several low-rank subspaces, as shown in Fig.\ref{fig:subspace}.
\begin{figure}
\begin{center}
\includegraphics[width=0.4\textwidth]{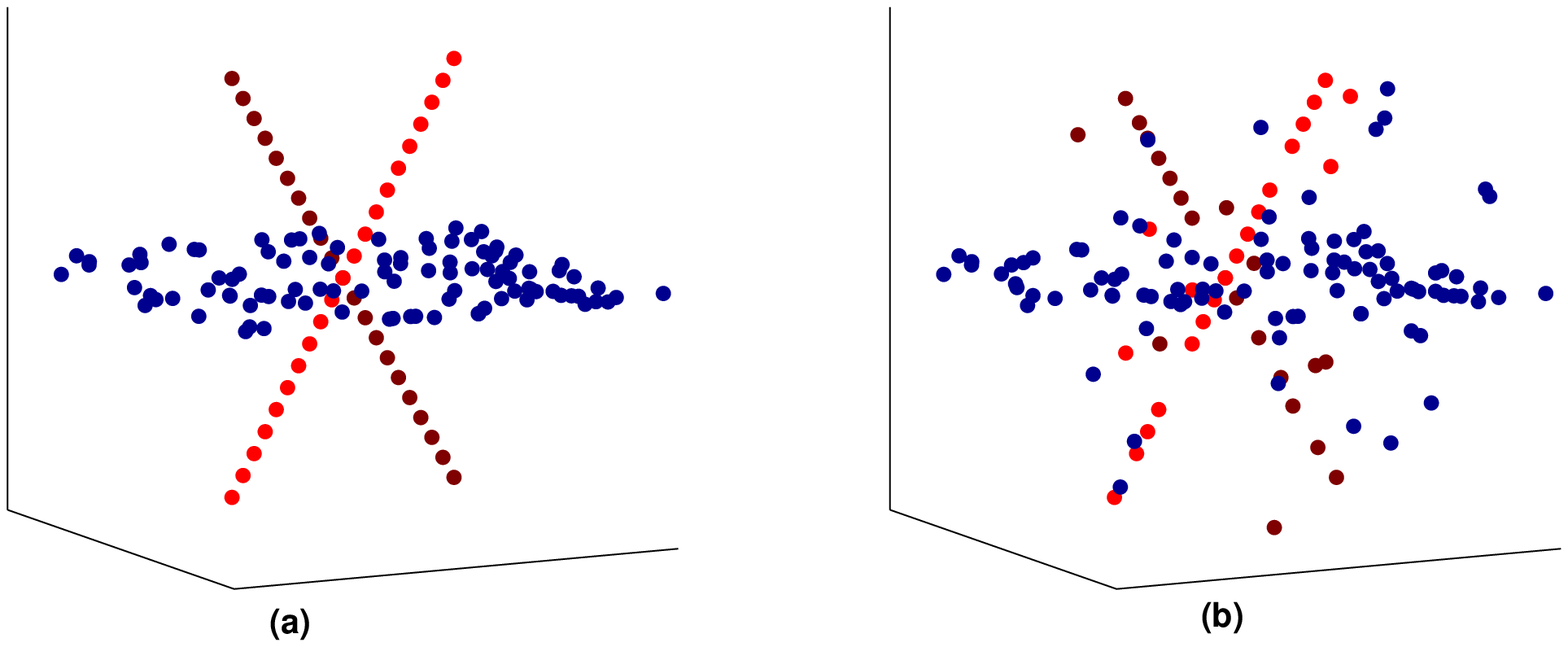}
\caption{\textbf{A mixture of subspaces consisting of a 2D plane and two 1D lines}. (a) The samples are strictly drawn from the underlying subspaces. (b) The samples are approximately drawn from the underlying subspaces.} \label{fig:subspace}\vspace{-0.25in}
\end{center}
\end{figure}

The generality and importance of subspaces naturally lead to a challenging problem of subspace segmentation (or clustering), whose goal is to segment (cluster or group) data into clusters with each cluster corresponding to a subspace. Subspace segmentation is an important data clustering problem and arises in numerous research areas, including computer vision \cite{motion_pami_2010_Rene,cvpr_2003_ksubspace,hgm}, image processing \cite{tpami_2008_acl,cacm_1981_ransac} and system identification \cite{subspace_system_identification}. When the data is clean, i.e., the samples are strictly drawn from the subspaces, several existing methods (e.g., \cite{ijcv_1998_factor,cvpr_2009_ssc,icml_2010_lrr}) are able to exactly solve the subspace segmentation problem. So, as pointed out by \cite{motion_pami_2010_Rene,icml_2010_lrr}, the main challenge of subspace segmentation is to handle the \emph{errors} (e.g., noise and corruptions) that possibly exist in data, i.e., to handle the data that may not strictly follow subspace structures. With this viewpoint, in this paper we therefore study the following \emph{subspace clustering} \cite{candes:2012:sr:outliers} problem.
\begin{problem}[Subspace Clustering]\label{pb:sp_recover}
Given a set of data samples \emph{approximately} (i.e., the data may contain errors) drawn from a union of linear subspaces, correct the possible errors and segment all samples into their respective subspaces simultaneously.
\end{problem}
\begin{figure}
\begin{center}
\includegraphics[width=0.45\textwidth]{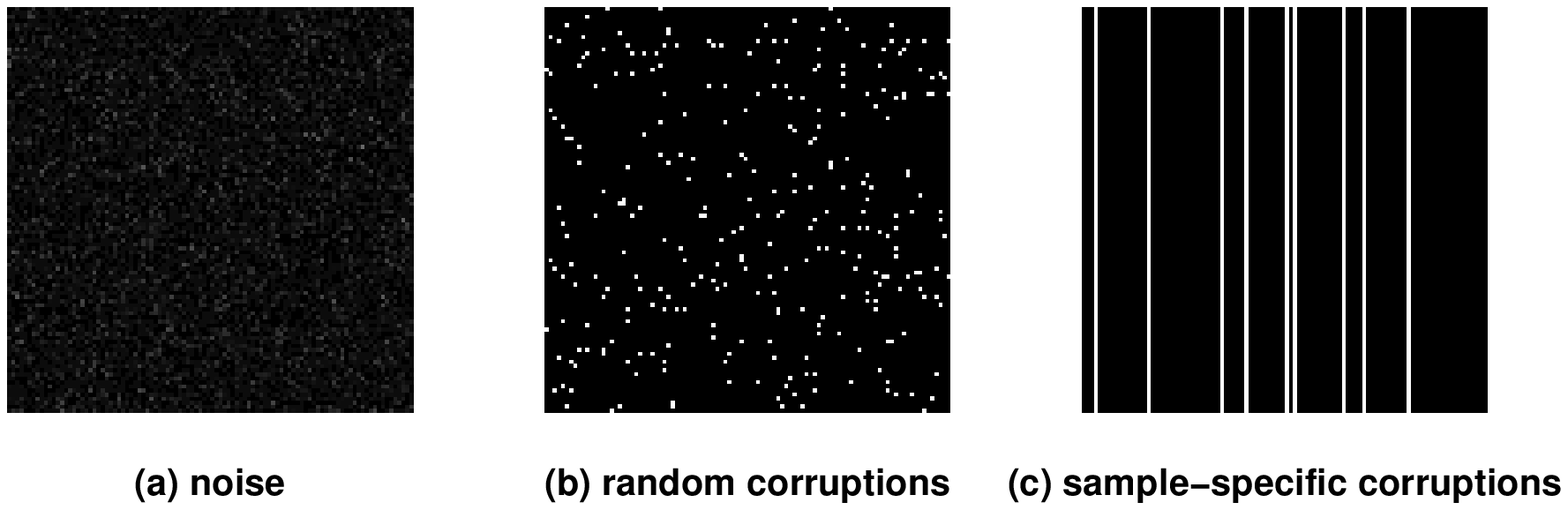}
\caption{\textbf{Illustrating three typical types of errors:} (a) noise \cite{CandesPIEEE}, which indicates the phenomena that the data is slightly perturbed around the subspaces (what we show is a perturbed data matrix whose columns are samples drawn from the subspaces); (b) random corruptions \cite{journal_2009_rpca2}, which indicate that a fraction of random entries are grossly corrupted; (c) sample-specific corruptions (and outliers), which indicate the phenomena that a fraction of the data samples (i.e., columns of the data matrix) are far away from the subspaces.} \label{fig:noise}\vspace{-0.25in}
\end{center}
\end{figure}

Notice that the word ``error'' generally refers to the \emph{deviation} between model assumption (i.e., subspaces) and data. It could exhibit as noise \cite{CandesPIEEE}, missed entries \cite{CandesPIEEE}, outliers \cite{xu:2010:nips} and corruptions \cite{journal_2009_rpca2} in reality. Fig.\ref{fig:noise} illustrates three typical types of errors under the context of subspace modeling. In this work, we shall focus on the sample-specific corruptions (and outliers) shown in Fig.\ref{fig:noise}(c), with mild concerns to the cases of Fig.\ref{fig:noise}(a) and Fig.\ref{fig:noise}(b). Notice that an outlier is from a different model other than subspaces, and is essentially different from a corrupted sample that belongs to the subspaces. We put them into the same category just because they can be handled in the same way, as will be shown in Section \ref{sec:assume2}.

To recover the subspace structures from the data containing errors, we propose a novel method termed \emph{low-rank representation} (LRR) \cite{icml_2010_lrr}. Given a set of data samples each of which can be represented as a linear combination of the bases in a dictionary, LRR aims at finding the \emph{lowest-rank} representation of all data jointly. The computational procedure of LRR is to solve a \emph{nuclear norm} \cite{phd_2002_nuclear} regularized optimization problem, which is convex and can be solved in polynomial time. By choosing a specific dictionary, it is shown that LRR can well solve the subspace clustering problem: when the data is clean, we prove that LRR exactly recovers the \emph{row space} of the data; for the data contaminated by outliers, we prove that under certain conditions LRR can \emph{exactly} recover the row space of the original data and detect the outlier as well; for the data corrupted by arbitrary errors, LRR can also approximately recover the row space with theoretical guarantees. Since the subspace membership is provably determined by the row space (we will discuss this in Section \ref{subsec:row_space}), these further imply that LRR can perform robust subspace clustering and error correction, in an efficient way. In summary, the contributions of this work include:
\begin{itemize}
\item[$\bullet$] We develop a simple yet effective method, termed LRR, which has been used to achieve state-of-the-art performance in several applications such as motion segmentation \cite{liu:2011:iccv}, image segmentation \cite{chen:2011:iccv}, saliency detection \cite{lang:2011:tip} and face recognition \cite{liu:2011:iccv}.
\item[$\bullet$] Our work extends the recovery of corrupted data from a single subspace \cite{journal_2009_rpca2} to multiple subspaces. Compared to \cite{corr_2008_union}, which requires the bases of subspaces to be known for handling the corrupted data from multiple subspaces, our method is \emph{autonomous}, i.e., no extra clean data is required.
\item[$\bullet$] Theoretical results for robust recovery are provided. While our analysis shares similar features as previous work in matrix completion \cite{CandesPIEEE} and robust PCA (RPCA) \cite{journal_2009_rpca2,xu:2010:nips}, it is considerably more challenging due to the fact that there is a dictionary matrix in LRR.
\end{itemize}
\vspace{-0.15in}\section{Related Work}\label{sec:related_work}
In this section, we discuss some existing subspace segmentation methods. In general, existing works can be roughly divided into four main categories: mixture of Gaussian, factorization, algebraic and spectral-type methods.

In statistical learning, mixed data is typically modeled as a set of independent samples drawn from a mixture of probabilistic distributions. As a single subspace can be well modeled by a (degenerate) Gaussian distribution, it is straightforward to assume that each probabilistic distribution is Gaussian, i.e., adopting a mixture of Gaussian models. Then the problem of segmenting the data is converted to a model estimation problem. The estimation can be performed either by using the Expectation Maximization (EM) algorithm to find a maximum likelihood estimate, as done in \cite{cvpr_2004_em}, or by iteratively finding a min-max estimate, as adopted by K-subspaces \cite{cvpr_2003_ksubspace} and Random Sample Consensus (RANSAC) \cite{cacm_1981_ransac}. These methods are sensitive to errors. So several efforts have been made for improving their robustness, e.g., the Median K-flats \cite{corr_2009_kflats} for K-subspaces, the work \cite{cvprw_2006_ransac} for RANSAC, and \cite{tpami_2008_acl} use a coding length to characterize a mixture of Gaussian. These refinements may introduce some robustness. Nevertheless, the problem is still not well solved due to the optimization difficulty, which is a bottleneck for these methods.

Factorization based methods \cite{ijcv_1998_factor} seek to approximate the given data matrix as a product of two matrices, such that the support pattern for one of the factors reveals the segmentation of the samples. In order to achieve robustness to noise, these methods modify the formulations by adding extra regularization terms. Nevertheless, such modifications usually lead to non-convex optimization problems, which need heuristic algorithms (often based on alternating minimization or EM-style algorithms) to solve. Getting stuck at local minima may undermine their performances, especially when the data is grossly corrupted. It will be shown that LRR can be regarded as a robust generalization of the method in \cite{ijcv_1998_factor} (which is referred to as PCA in this paper). The formulation of LRR is convex and can be solved in polynomial time.

Generalized Principal Component Analysis (GPCA) \cite{siam_2008_gpca} presents an algebraic way to model the data drawn from a union of multiple subspaces. This method describes a subspace containing a data point by using the gradient of a polynomial at that point. Then subspace segmentation is made equivalent to fitting the data with polynomials. GPCA can guarantee the success of the segmentation under certain conditions, and it does not impose any restriction on the subspaces. However, this method is sensitive to noise due to the difficulty of estimating the polynomials from real data, which also causes the high computation cost of GPCA. Recently, Robust Algebraic Segmentation (RAS) \cite{ijcv_2009_ras} has been proposed to resolve the robustness issue of GPCA. However, the computation difficulty for fitting polynomials is unfathomably large. So RAS can make sense only when the data dimension is low and the number of subspaces is small.

As a data clustering problem, subspace segmentation can be done by firstly learning an affinity matrix from the given data, and then obtaining the final segmentation results by spectral clustering algorithms such as Normalized Cuts (NCut) \cite{tpami_2000_ncut}. Many existing methods such as Sparse Subspace Clustering (SSC) \cite{cvpr_2009_ssc}, Spectral Curvature Clustering (SCC) \cite{Chen:2009:SCC,Chen:2009:SCC:theory}, Spectral Local Best-fit Flats (SLBF) \cite{Zhang:2011:lbf,Lerman:2011:LLA}, the proposed LRR method and \cite{eccv_2006_lsa,sc:iccv:2009}, possess such spectral nature, so called as spectral-type methods. The main difference among various spectral-type methods is the approach for learning the affinity matrix. Under the assumption that the data is clean and the subspaces are independent, \cite{cvpr_2009_ssc} shows that solution produced by sparse representation (SR) \cite{donoho:2004:sr} could achieve the so-called $\ell_1$ Subspace Detection Property ($\ell_1$-SDP): the within-class affinities are sparse and the between-class affinities are all zeros. In the presence of outliers, it is shown in \cite{candes:2012:sr:outliers} that the SR method can still obey $\ell_1$-SDP. However, $\ell_1$-SDP may not be sufficient to ensure the success of subspace segmentation \cite{Behrooz:cvpr:2011}. Recently, Lerman and Zhang \cite{Lerman:2011:Lp} prove that under certain conditions the multiple subspace structures can be exactly recovered via $\ell_p$ ($p\leq1$) minimization. Unfortunately, since the formulation is not convex, it is still unknown how to efficiently obtain the globally optimal solution. In contrast, the formulation of LRR is convex and the corresponding optimization problem can be solved in polynomial time. What is more, even if the data is contaminated by outliers, the proposed LRR method is proven to exactly recover the right row space, which provably determines the subspace segmentation results (we shall discuss this in Section \ref{subsec:row_space}). In the presence of arbitrary errors (e.g., corruptions, outliers and noise), LRR is also guaranteed to produce near recovery.
\vspace{-0.15in}\section{Preliminaries and Problem Statement}\label{sec:notation}
\subsection{Summary of Main Notations}
In this work, matrices are represented with capital symbols. In particular, $\Id$ is used to denote the identity matrix, and the entries of matrices are denoted by using $[\cdot]$ with subscripts. For instance, $M$ is a matrix, $[M]_{ij}$ is its $(i,j)$-th entry, $[M]_{i,:}$ is its $i$-th row, and $[M]_{:,j}$ is its $j$-th column. For ease of presentation, the horizontal (resp. vertical) concatenation of a collection of matrices along row (resp. column) is denoted by $[M_1,M_2,\cdots,M_k]$ (resp. $[M_1;M_2;\cdots;M_k]$). The block-diagonal matrix formed by a collection of matrices $M_1,M_2,\cdots,M_k$ is denoted by
\begin{eqnarray}\label{eq:blk}
\diag{M_1,M_2,\cdots,M_k}=\left[\begin{array}{cccc}
M_1&0&0&0\\
0&M_2&0&0\\
0&0&\ddots&0\\
0&0&0&M_k\\
\end{array}\right].
\end{eqnarray}

The only used vector norm is the $\ell_2$ norm, denoted by $\norm{\cdot}_2$. A variety of norms on matrices will be used. The matrix $\ell_0$, $\ell_{2,0}$,
$\ell_1$, $\ell_{2,1}$ norms are defined by $\norm{M}_0=\#\{(i,j):[M]_{ij}\neq0\}$, $\norm{M}_{2,0}=\#\{i:\|[M]_{:,i}\|_2\neq0\}$, $\norm{M}_1=\sum_{i,j}|[M]_{ij}|$ and $\norm{M}_{2,1}=\sum_{i}\|[M]_{:,i}\|_2$, respectively. The matrix $\ell_{\infty}$ norm is defined as $\norm{M}_{\infty}=\max_{i,j}|[M]_{ij}|$. The spectral norm of a matrix $M$ is denoted by $\norm{M}$, i.e., $\norm{M}$ is the largest singular value of $M$. The Frobenius norm and the nuclear norm (the sum of singular values of a matrix) are denoted by $\norm{M}_F$ and $\norm{M}_*$, respectively. The Euclidean inner product between two matrices is $\langle{}M,N\rangle=\trace{M^TN}$, where $M^T$ is the transpose of a matrix and $\trace{\cdot}$ is the trace of a matrix.

The supports of a matrix $M$ are the indices of its nonzero entries, i.e., $\{(i,j):[M]_{ij}\neq0\}$. Similarly, its column supports are the indices of its nonzero columns. The symbol $\mathcal{I}$ (superscripts, subscripts, etc.) is used to denote the column supports of a matrix, i.e., $\mathcal{I}=\{(i):\|[M]_{:,i}\|_2\neq0\}$. The corresponding complement set (i.e., zero columns) is $\mathcal{I}^c$. There are two projection operators associated with $\mathcal{I}$ and $\mathcal{I}^c$: $\mathcal{P}_{\mathcal{I}}$ and $\mathcal{P}_{\mathcal{I}^c}$. While applying them to a matrix $M$, the matrix $\mathcal{P}_{\mathcal{I}}(M)$ (resp. $\mathcal{P}_{\mathcal{I}^c}(M)$) is obtained from $M$ by setting $[M]_{:,i}$ to zero for all $i\not\in{}\mathcal{I}$ (resp. $i\not\in{}\mathcal{I}^c$).

We also adopt the conventions of using $\spn{M}$ to denote the linear space spanned by the columns of a matrix $M$, using $y\in\spn{M}$ to denote that a vector $y$ belongs to the space $\spn{M}$, and using $Y\in\spn{M}$ to denote that all column vectors of $Y$ belong to $\spn{M}$.

Finally, in this paper we use several terminologies, including ``block-diagonal matrix'', ``union and sum of subspaces'', ``independent (and disjoint) subspaces'', ``full SVD and skinny SVD'', ``pseudoinverse'', ``column space and row space'' and ``affinity degree''. These terminologies are defined in Appendix.
\vspace{-0.1in}\subsection{Relations Between Segmentation and Row Space}\label{subsec:row_space}
\begin{figure}
\begin{center}
\includegraphics[width=0.2\textwidth]{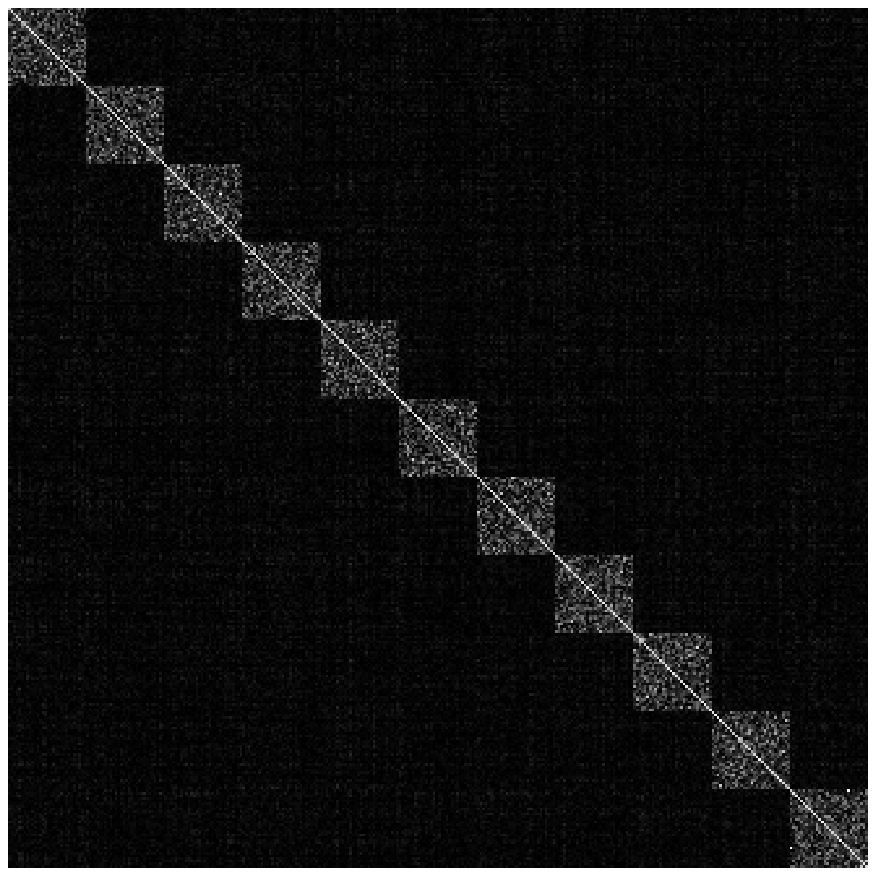}
\caption{\textbf{An example of the matrix $V_0V_0^T$ computed from dependent subspaces.} In this example, we create 11 pairwise disjoint subspaces each of which is of dimension 20, and draw 20 samples from each subspace without errors. The ambient dimension is 200, which is smaller than the sum of the dimensions of the subspaces. So the subspaces are dependent and $V_0V_0^T$ is not strictly block-diagonal. Nevertheless, it is simple to see that high segmentation accuracy can be achieved by using the above affinity matrix to do spectral clustering.}\vspace{-0.25in}\label{fig:sim}
\end{center}
\end{figure}
Let $X_0$ with skinny SVD $U_0\Sigma_0V_0^T$ be a collection of data samples \emph{strictly} drawn from a union of multiple subspaces (i.e., $X_0$ is clean), the subspace membership of the samples is determined by the row space of $X_0$. Indeed, as shown in \cite{ijcv_1998_factor}, when subspaces are independent, $V_0V_0^T$ forms a block-diagonal matrix: the $(i,j)$-th entry of $V_0V_0^T$ can be nonzero only if the $i$-th and $j$-th samples are from the same subspace. Hence, this matrix, termed as Shape Interaction Matrix (SIM) \cite{ijcv_1998_factor}, has been widely used for subspace segmentation. Previous approaches simply compute the SVD of the data matrix $X=U_X\Sigma_{X}V_X^T$ and then use $|V_XV_X^T|$ \footnote{For a matrix $M$, $|M|$ denotes the matrix with the $(i,j)$-th entry being the absolute value of $[M]_{ij}$.} for subspace segmentation. However, in the presence of outliers and corruptions, $V_X$ can be far away from $V_0$ and thus the segmentation using such approaches is inaccurate. In contrast, we show that LRR can recover $V_0V_0^T$ even when the data matrix $X$ is contaminated by outliers.

If the subspaces are not independent, $V_0V_0^T$ may not be strictly block-diagonal. This is indeed well expected, since when the subspaces have nonzero (nonempty) intersections, then some samples may belong to multiple subspaces simultaneously. When the subspaces are pairwise disjoint (but not independent), our extensive numerical experiments show that $V_0V_0^T$ may still be close to be block-diagonal, as exemplified in Fig. \ref{fig:sim}. Hence, to recover $V_0V_0^T$ is still of interest to subspace segmentation.
\vspace{-0.1in}\subsection{Problem Statement}
Problem \ref{pb:sp_recover} only roughly describes what we want to study. More precisely, this paper addresses the following problem.
\begin{problem}[Subspace Clustering]\label{pb:sp_recover_precise}
Let $X_0\in{}\mathbb{R}^{d\times{}n}$ with skinny SVD $U_0\Sigma_0V_0$ store a set of $n$ $d$-dimensional samples (vectors) strictly drawn from a union of $k$ subspaces $\{\mathcal{S}_i\}_{i=1}^{k}$ of unknown dimensions ($k$ is unknown either). Given a set of observation vectors $X$ generated by
\begin{eqnarray*}
X=X_0+E_0,
\end{eqnarray*}
the goal is to recover the row space of $X_0$, or to recover the true SIM $V_0V_0^T$ as equal.
\end{problem}

The recovery of row space can guarantee high segmentation accuracy, as analyzed in Section \ref{subsec:row_space}. Also, the recovery of row space naturally implies the success in error correction. So it is sufficient to set the goal of subspace clustering as the recovery of the row space identified by $V_0V_0^T$. For ease of exploration, we consider the problem under three assumptions of increasing practicality and difficulty.
\subsubsection*{Assumption 1} The data is clean, i.e., $E_0=0$.
\subsubsection*{Assumption 2} A fraction of the data samples are grossly corrupted and the others are clean, i.e., $E_0$ has sparse column supports as shown in Fig.\ref{fig:noise}(c).
\subsubsection*{Assumption 3} A fraction of the data samples are grossly corrupted and the others are contaminated by small Gaussian noise, i.e., $E_0$ is characterized by a combination of the models shown in Fig.\ref{fig:noise}(a) and Fig.\ref{fig:noise}(c).

Unlike \cite{icml_2010_lrr}, the independent assumption on the subspaces is not highlighted in this paper, because the analysis in this work focuses on recovering $V_0V_0^T$ other than a pursuit of block-diagonal
matrix.
\vspace{-0.15in}\section{Low-Rank Representation for Matrix Recovery}\label{sec:lrr}
In this section we abstractly present the LRR method for recovering a matrix from corrupted observations. The basic theorems and optimization algorithms will be presented. The specific methods and theories for handling the subspace clustering problem are deferred until Section \ref{sec:subrecovery}.
\vspace{-0.1in}\subsection{Low-Rank Representation}\label{sec:general_results:formulation}
In order to recover the low-rank matrix $X_0$ from the given observation matrix $X$ corrupted by errors $E_0$ ($X=X_0+E_0$), it is straightforward to consider the following regularized rank minimization problem:
\begin{eqnarray}\label{eq:rpca}
\min_{D,E} \rank{D}+\lambda\norm{E}_{\ell}, & \textrm{s.t.}&X=D+E,
\end{eqnarray}
where $\lambda>0$ is a parameter and $\norm{\cdot}_{\ell}$ indicates certain regularization strategy, such as the squared Frobenius norm (i.e., $\|\cdot\|_F^2$) used for modeling the noise as show in Fig.\ref{fig:noise}(a) \cite{CandesPIEEE}, the $\ell_0$ norm adopted by \cite{journal_2009_rpca2} for characterizing the random corruptions as shown in Fig.\ref{fig:noise}(b), and the $\ell_{2,0}$ norm adopted by \cite{icml_2010_lrr,xu:2010:nips} for dealing with sample-specific corruptions and outliers. Suppose $D^*$ is a minimizer with respect to the variable $D$, then it gives a low-rank recovery to the original data $X_0$.

The above formulation is adopted by the recently established Robust PCA (RPCA) method \cite{journal_2009_rpca2} which has been used to achieve the state-of-the-art performance in several applications (e.g., \cite{tilt:accv:2010}). However, this formulation implicitly assumes that the underlying data structure is a single low-rank subspace. When the data is drawn from a union of multiple subspaces, denoted as $\mathcal{S}_1,\mathcal{S}_2,\cdots,\mathcal{S}_k$, it actually treats the data as being sampled from a single subspace defined by $\mathcal{S}=\sum_{i=1}^{k}\mathcal{S}_i$. Since the sum $\sum_{i=1}^{k}\mathcal{S}_i$ can be much larger than the union $\cup_{i=1}^{k}\mathcal{S}_i$, the specifics of the individual subspaces are not well considered and so the recovery may be inaccurate.

To better handle the mixed data, here we suggest a more general rank minimization problem defined as follows:
\begin{eqnarray}\label{eq:general:lrr}
\min_{Z,E} \rank{Z}+\lambda\norm{E}_{\ell},& \textrm{s.t.}& X=AZ+E,
\end{eqnarray}
where $A$ is a ``dictionary'' that linearly spans the data space. We call the minimizer $Z^*$ (with regard to the variable $Z$) the ``lowest-rank representation'' of data $X$ with respect to a dictionary $A$. After obtaining an optimal solution $(Z^*,E^*)$, we could recover the original data by using $AZ^*$ (or $X-E^*$). Since $\rank{AZ^*}\leq\rank{Z^*}$, $AZ^*$ is also a low-rank recovery to the original data $X_0$. By setting $A=\Id$, the formulation \eqref{eq:general:lrr} falls back to \eqref{eq:rpca}. So LRR could be regarded as a generalization of RPCA that essentially uses the standard bases as the dictionary.
By choosing an appropriate dictionary $A$, as we will see, the lowest-rank representation can recover the underlying row space so as to reveal the true segmentation of data. So, LRR could handle well the data drawn from a union of multiple subspaces.
\vspace{-0.1in}\subsection{Analysis on the LRR Problem}
The optimization problem \eqref{eq:general:lrr} is difficult to solve due to the discrete nature of the rank function. For ease of exploration, we begin with the ``ideal'' case that the data is clean. That is, we consider the following rank minimization problem:
\begin{eqnarray}\label{eq:lrr:low_rank}
\min_{Z} \rank{Z}, & \textrm{s.t.} & X=AZ.
\end{eqnarray}
It is easy to see that the solution to \eqref{eq:lrr:low_rank} may not be unique. As a common practice in rank minimization problems, we replace the rank function with the nuclear norm, resulting in the following convex optimization problem:
\begin{eqnarray}\label{eq:lr:nuclear_norm_minization}
\min_{Z} \norm{Z}_*, & \textrm{s.t.} & X=AZ.
\end{eqnarray}
We will show that the solution to \eqref{eq:lr:nuclear_norm_minization} is also a solution to \eqref{eq:lrr:low_rank} and this special solution is useful for subspace segmentation.

In the following, we shall show some general properties of the minimizer to problem \eqref{eq:lr:nuclear_norm_minization}. These general conclusions form the foundations of LRR (the proofs can be found in Appendix).
\subsubsection{Uniqueness of the Minimizer}The nuclear norm is convex, but not strongly convex. So it is possible that problem \eqref{eq:lr:nuclear_norm_minization} has multiple optimal solutions. Fortunately, it can be proven that the minimizer to problem \eqref{eq:lr:nuclear_norm_minization} is \emph{always} uniquely defined by a closed form. This is summarized in the following theorem.
\begin{theorem}\label{theorem:unique:nonoise}
Assume $A\neq0$ and $X=AZ$ have feasible solution(s), i.e., $X\in\spn{A}$. Then
\begin{eqnarray}\label{eq:lrr:close_form}
Z^*=A^{\dag}X,
\end{eqnarray}
is the unique minimizer to problem \eqref{eq:lr:nuclear_norm_minization}, where $A^{\dag}$ is the pseudoinverse of $A$.
\end{theorem}

From the above theorem, we have the following corollary which shows that problem \eqref{eq:lr:nuclear_norm_minization} is a good surrogate of problem \eqref{eq:lrr:low_rank}.
\begin{corollary}\label{corollary:equal}
Assume $A\neq0$ and $X=AZ$ have feasible solutions. Let $Z^*$ be the minimizer to problem
\eqref{eq:lr:nuclear_norm_minization}, then $\rank{Z^*}=\rank{X}$ and $Z^*$ is also a minimal rank solution to
problem \eqref{eq:lrr:low_rank}.
\end{corollary}
\subsubsection{Block-Diagonal Property of the Minimizer} By choosing an appropriate dictionary, the lowest-rank representation can reveal the true segmentation results. Namely, when the columns of $A$ and $X$ are exactly sampled from independent subspaces, the minimizer to problem \eqref{eq:lr:nuclear_norm_minization} can reveal the subspace membership among the samples. Let $\{\mathcal{S}_1,\mathcal{S}_2,\cdots,\mathcal{S}_k\}$ be a collection of $k$ subspaces, each of which has a rank (dimension) of $r_i>0$. Also, let $A=[A_1,A_2,\cdots,A_k]$ and $X=[X_1,X_2,\cdots,X_k]$. Then we have the following theorem.
\begin{theorem}\label{theorem:exact:nonoise}
Without loss of generality, assume that $A_i$ is a collection of $m_i$ samples of the $i$-th subspace $\mathcal{S}_i$, $X_i$ is a collection of $n_i$
samples from $\mathcal{S}_i$, and the sampling of each $A_i$ is sufficient such that $\rank{A_i}=r_i$ (i.e., $A_i$ can be regarded as the bases that span the subspace). If the subspaces are independent, then
the minimizer to problem \eqref{eq:lr:nuclear_norm_minization} is block-diagonal:
\begin{eqnarray*}
Z^*=\left[\begin{array}{cccc}
Z_1^*&0&0&0\\
0&Z_2^*&0&0\\
0&0&\ddots&0\\
0&0&0&Z_k^*
\end{array}\right],
\end{eqnarray*}
where $Z_i^*$ is an $m_i\times{n_i}$ coefficient matrix with $\rank{Z_i^*} = \rank{X_i}, \; \forall \, i.$
\end{theorem}

Note that the claim of $\rank{Z_i^*} = \rank{X_i}$ guarantees the high within-class homogeneity of $Z_i^*$, since the low-rank properties generally requires $Z_i^*$ to be dense. This is different from SR, which is prone to produce a ``trivial'' solution if $A=X$, because the sparsest representation is an identity matrix in this case. It is also worth noting that the above block-diagonal property does not require the data samples have been grouped together according to their subspace memberships. There is no loss of generality to assume that the indices of the samples have been rearranged to satisfy the true subspace memberships, because the solution produced by LRR is globally optimal and does not depend on the arrangements of the data samples.
\vspace{-0.1in}\subsection{Recovering Low-Rank Matrices by Convex Optimization}
Corollary \ref{corollary:equal} suggests that it is appropriate to use the nuclear norm as a surrogate to replace
the rank function in problem \eqref{eq:general:lrr}. Also, the matrix $\ell_1$ and $\ell_{2,1}$ norms are good relaxations of the $\ell_0$ and $\ell_{2,0}$ norms, respectively.
So we could obtain a low-rank recovery to $X_0$ by solving the following convex optimization problem:
\begin{eqnarray}\label{eq:lrr:l1_l2_regular}
\min_{Z,E} \norm{Z}_*+\lambda{\norm{E}_{2,1}}, &\textrm{s.t.} & X = AZ+E.
\end{eqnarray}
Here, the $\ell_{2,1}$ norm is adopted to characterize the error term $E$, since we want to model the sample-specific corruptions (and outliers) as shown in Fig.\ref{fig:noise}(c). For the small Gaussian noise as shown in Fig.\ref{fig:noise}(a), $\|E\|_F^2$ should be chosen; for the random corruptions as shown in Fig.\ref{fig:noise}(b), $\|E\|_1$ is an appropriate choice. After obtaining the minimizer $(Z^*,E^*)$, we could use $AZ^*$ (or $X-E^*$) to obtain a low-rank recovery to the original data $X_0$.
\begin{algorithm}[tb]
   \caption{Solving Problem \eqref{eq:lrr:l1_l2_regular} by Inexact ALM}
   \label{alg:alm}
\begin{algorithmic}
   \STATE {\bfseries Input:} data matrix $X$, parameter $\lambda$.
   \STATE {\bfseries Initialize:} $Z=J=0,E=0,Y_1=0,Y_2=0,\mu=10^{-6},\mu_{max}=10^{6},\rho=1.1$, and $\varepsilon=10^{-8}$.
   \WHILE{not converged}
   \STATE \textbf{1.} fix the others and update $J$ by $$J=\arg\min\frac{1}{\mu}||J||_*+\frac{1}{2}||J-(Z+Y_2/\mu)||_F^2.$$
   \STATE \textbf{2.} fix the others and update $Z$ by
   $$Z=(\Id+A^TA)^{-1}(A^T(X-E)+J+(A^TY_1-Y_2)/\mu).$$
   \STATE \textbf{3.} fix the others and update $E$ by $$E = \arg\min\frac{\lambda}{\mu}||E||_{2,1}+\frac{1}{2}||E-(X-AZ+Y_1/\mu)||_F^2.$$
   \STATE \textbf{4.} update the multipliers
   \begin{eqnarray*}
   Y_1 &=& Y_1 + \mu(X-AZ-E),\\
   Y_2 &=& Y_2 + \mu(Z-J).
   \end{eqnarray*}
   \STATE \textbf{5.} update the parameter $\mu$ by
   $\mu=\min(\rho\mu,\mu_{max})$.
   \STATE \textbf{6.} check the convergence conditions:
   $$||X-AZ-E||_{\infty}<\varepsilon \textrm{ and }  ||Z-J||_{\infty}<\varepsilon.$$
   \ENDWHILE
\end{algorithmic}
\end{algorithm}

The optimization problem \eqref{eq:lrr:l1_l2_regular} is convex and can be solved by various methods. For efficiency, we adopt in this paper the Augmented Lagrange Multiplier (ALM) \cite{lin_alm,alm:1980} method.  We first convert \eqref{eq:lrr:l1_l2_regular} to the following equivalent problem:
\begin{eqnarray*}
\min_{Z,E,J} \norm{J}_*+\lambda{\norm{E}_{2,1}},\textrm{ s.t. } X = AZ+E, Z=J.
\end{eqnarray*}
This problem can be solved by the ALM method, which minimizes the following augmented Lagrange function:
\begin{eqnarray*}
&\mathcal{L}=\norm{J}_*+\lambda{\norm{E}_{2,1}} + \trace{Y_1^T(X-AZ-E)}+\\
&\trace{Y_2^T(Z-J)}
+\frac{\mu}{2}(\norm{X-AZ-E}_F^2+\norm{Z-J}_F^2).
\end{eqnarray*}
The above problem is unconstrained. So it can be minimized with respect to $J$, $Z$ and $E$, respectively, by fixing the other variables, and then updating the Lagrange multipliers $Y_1$ and $Y_2$, where $\mu>0$ is a penalty parameter. The inexact ALM method, also called the alternating direction method, is outlined in Algorithm \ref{alg:alm} \footnote{To solve the problem $
\min_{Z,E} \norm{Z}_*+\lambda{\norm{E}_{1}}, \textrm{ s.t. } X = AZ+E$, one only needs to replace Step 3 of Algorithm \ref{alg:alm} by $E = \arg\min\frac{\lambda}{\mu}||E||_{1}+\frac{1}{2}||E-(X-AZ+Y_1/\mu)||_F^2$, which is solved by using the shrinkage operator \cite{lin_alm}.

Also, please note here that the setting of $\varepsilon=10^{-8}$ is based on the assumption that the values in $X$ has been normalized within the range of $0\sim1$.}. Note that although Step 1 and Step 3 of the algorithm are convex problems, they both have closed-form solutions. Step 1 is solved via the Singular Value Thresholding (SVT) operator \cite{svt:cai:2008}, while Step 3 is solved via the following lemma:
\begin{lemma}[\cite{solve_l2l1:siam:2009}]\label{lemma:solve_l2l1} Let $Q$ be a given matrix. If the optimal solution to
\begin{eqnarray*}
\min_{W} \alpha||W||_{2,1} + \frac{1}{2}||W-Q||_F^2
\end{eqnarray*}
is $W^*$, then the $i$-th column of $W^*$ is
\begin{eqnarray*}
[W^*]_{:,i}=\left\{
\begin{array}{ll} \frac{||Q_{:,i}||_2-\alpha}{||Q_{:,i}||_2}Q_{:,i}, & \mbox{if $||Q_{:,i}||_2>\alpha$};\\
0, & \mbox{otherwise.}
\end{array}\right.
\end{eqnarray*}
\end{lemma}
\subsubsection{Convergence Properties} When the objective function is smooth, the convergence of the exact ALM algorithm has been generally proven in \cite{alm:1980}. For inexact ALM, which is a variation of exact ALM, its convergence has also been well studied when the number of blocks is at most two \cite{lin_alm,zhang:2010:adm}. Up to present, it is still difficult to \emph{generally} ensure the convergence of inexact ALM with three or more blocks \cite{zhang:2010:adm}. Since there are three blocks (including $Z,J$ and $E$) in Algorithm \ref{alg:alm} and the objective function of \eqref{eq:lrr:l1_l2_regular} is not smooth, it would be not easy to prove the convergence in theory.

Fortunately, there actually exist some guarantees for ensuring the convergence of Algorithm \ref{alg:alm}. According to the theoretical results in \cite{Eckstein:1992:DSM}, two conditions are \emph{sufficient} (but may not necessary) for Algorithm \ref{alg:alm} to converge: the first condition is that the dictionary matrix $A$ is of full column rank; the second one is that the optimality gap produced in each iteration step is monotonically decreasing, namely the error
\begin{eqnarray*}
\epsilon_k=\|(Z_k,J_k)-\arg\min_{Z,J}\mathcal{L}\|_F^2
\end{eqnarray*}
is monotonically decreasing, where $Z_k$ (resp. $J_k$) denotes the solution produced at the $k$-th iteration, $\arg\min_{Z,J}\mathcal{L}$ indicates the ``ideal'' solution obtained by minimizing the Lagrange function $\mathcal{L}$ with respect to both $Z$ and $J$ simultaneously. The first condition is easy to obey, since problem \eqref{eq:lrr:l1_l2_regular} can be converted into an equivalent problem where the full column rank condition is always satisfied (we will show this in the next subsection). For the monotonically decreasing condition, although it is not easy to \emph{strictly} prove it, the convexity of the Lagrange function could guarantee its validity to some extent \cite{Eckstein:1992:DSM}. So, it could be well expected that Algorithm \ref{alg:alm} has good convergence properties. Moreover, inexact ALM is known to \emph{generally} perform well in reality, as illustrated in \cite{zhang:2010:adm}.

That $\mu$ should be upper bounded (Step 5 of Algorithm \ref{alg:alm}) is required by the traditional theory of the alternating direction method in order to guarantee the convergence of the algorithm. So we also adopt this convention. Nevertheless, please note that the upper boundedness may not be necessary for some particular problems, e.g., the RPCA problem as analyzed in \cite{lin_alm}.
\subsubsection{Computational Complexity}\label{sec:complexity} For ease of analysis, we assume that the sizes of both $A$ and $X$ are $d\times{}n$ in the following. The major computation of Algorithm \ref{alg:alm} is Step 1, which requires computing the SVD of an $n\times{}n$ matrix. So it will be time consuming if $n$ is large, i.e., the number of data samples is large. Fortunately, the computational cost of LRR can be easily reduced by the following theorem, which is followed from Theorem \ref{theorem:unique:nonoise}.
\begin{theorem}\label{the:lrr:solution}
For any optimal solution $(Z^*,E^*)$ to the LRR problem \eqref{eq:lrr:l1_l2_regular}, we have that $$Z^*\in{}\spn{A^T}.$$
\end{theorem}

The above theorem concludes that the optimal solution $Z^*$ (with respect to the variable $Z$) to \eqref{eq:lrr:l1_l2_regular} always lies within the subspace spanned by the rows of $A$. This means that $Z^*$ can be factorized into $Z^* = P^*\tilde{Z}^*$, where $P^*$ can be computed in advance by orthogonalizing the columns of $A^T$. Hence, problem \eqref{eq:lrr:l1_l2_regular} can be equivalently transformed into a simpler problem by replacing $Z$ with $P^*\tilde{Z}$:
\begin{eqnarray*}
\min_{\tilde{Z},E}\|\tilde{Z}\|_*+\lambda\norm{E}_{2,1},\textrm{ s.t. }X=B\tilde{Z}+E,
\end{eqnarray*}
where $B=AP^*$. After obtaining a solution $(\tilde{Z}^*,E^*)$ to the above problem, the optimal solution to \eqref{eq:lrr:l1_l2_regular} is recovered by $(P^*\tilde{Z}^*,E^*)$. Since the number of rows of $\tilde{Z}$ is at most $r_A$ (the rank of $A$), the above problem can be solved with a complexity of $O(dnr_A+nr_A^2+r_A^3)$ by using Algorithm \ref{alg:alm}. So LRR is quite scalable for large-size ($n$ is large) datasets, provided that a low-rank dictionary $A$ has been obtained. While using $A=X$, the computational complexity is at most $O(d^2n+d^3)$ (assuming $d\leq{}n$). This is also fast provided that the data dimension $d$ is not high.

While considering the cost of orthogonalization and the number of iterations needed to converge, the complexity of Algorithm \ref{alg:alm} is $$O(d^2n)+O(n_s(dnr_A+nr_A^2+r_A^3)),$$ where $n_s$ is the number of iterations. The iteration number $n_s$ depends on the choice of $\rho$: $n_s$ is smaller while $\rho$ is larger, and vice versa. Although larger $\rho$ does produce higher efficiency, it has the risk of losing optimality to use large $\rho$ \cite{lin_alm}. In our experiments, we always set $\rho=1.1$. Under this setting, the iteration number usually locates within the range of $50\sim300$.
\vspace{-0.15in}\section{Subspace Clustering by LRR}\label{sec:subrecovery}
In this section, we utilize LRR to address Problem \ref{pb:sp_recover_precise}, which is to recover the original row space from a set of corrupted observations. Both theoretical and experimental results will be presented.
\vspace{-0.1in}\subsection{Exactness to Clean Data}
When there are no errors in data, i.e., $X=X_0$ and $E_0=0$, it is simple to show that the row space (identified by $V_0V_0^T$) of $X_0$ is exactly recovered by solving the following nuclear norm minimization problem:
\begin{eqnarray}\label{eq:lrrx:assu1}
\min_{Z}\norm{Z}_*,&\textrm{s.t.}&X=XZ,
\end{eqnarray}
which is to choose the data matrix $X$ itself as the dictionary in \eqref{eq:lr:nuclear_norm_minization}. By Theorem \ref{theorem:unique:nonoise}, we have the following theorem which has also been proven by Wei and Lin \cite{rsi:tsp}.
\begin{theorem}
Suppose the skinny SVD of $X$ is $U\Sigma{}V^T$, then the minimizer to problem \eqref{eq:lrrx:assu1} is uniquely defined by $$Z^*=VV^T.$$ This naturally implies that $Z^*$ exactly recovers $V_0V_0^T$ when $X$ is clean (i.e., $E_0=0$).
\end{theorem}

The above theorem reveals the connection between LRR and the method in \cite{ijcv_1998_factor}, which is a counterpart of PCA (referred to as ``PCA'' for simplicity). Nevertheless, it is well known that PCA is fragile to the presence of outliers. In contrast, it can be proven in theory that LRR exactly recovers the row space of $X_0$ from the data contaminated by outliers, as will be shown in the next subsection.
\vspace{-0.1in}\subsection{Robustness to Outliers and Sample-Specific Corruptions}\label{sec:assume2}
Assumption 2 is to imagine that a fraction of the data samples are away from the underlying subspaces. This implies that the error term $E_0$ has sparse column supports. So, the $\ell_{2,1}$ norm is appropriate for characterizing $E_0$. By choosing $A=X$ in \eqref{eq:lrr:l1_l2_regular}, we have the following convex optimization problem:
\begin{eqnarray}\label{eq:lrrx:assu2}
\min_{Z,E} ||Z||_*+\lambda{||E||_{2,1}}, & \textrm{s.t.} & X = XZ+E.
\end{eqnarray}

The above formulation ``seems'' questionable, because the data matrix (which itself can contain errors) is used as the dictionary for error correction. Nevertheless, as shown in the following two subsections, $A=X$ is indeed a good choice for several particular problems \footnote{Note that this does not deny the importance of learning the dictionary. Indeed, the choice of dictionary is a very important aspect in LRR. We leave this as future work.}.
\subsubsection{Exactness to Outliers}
\begin{figure}
\begin{center}
\includegraphics[width=0.4\textwidth]{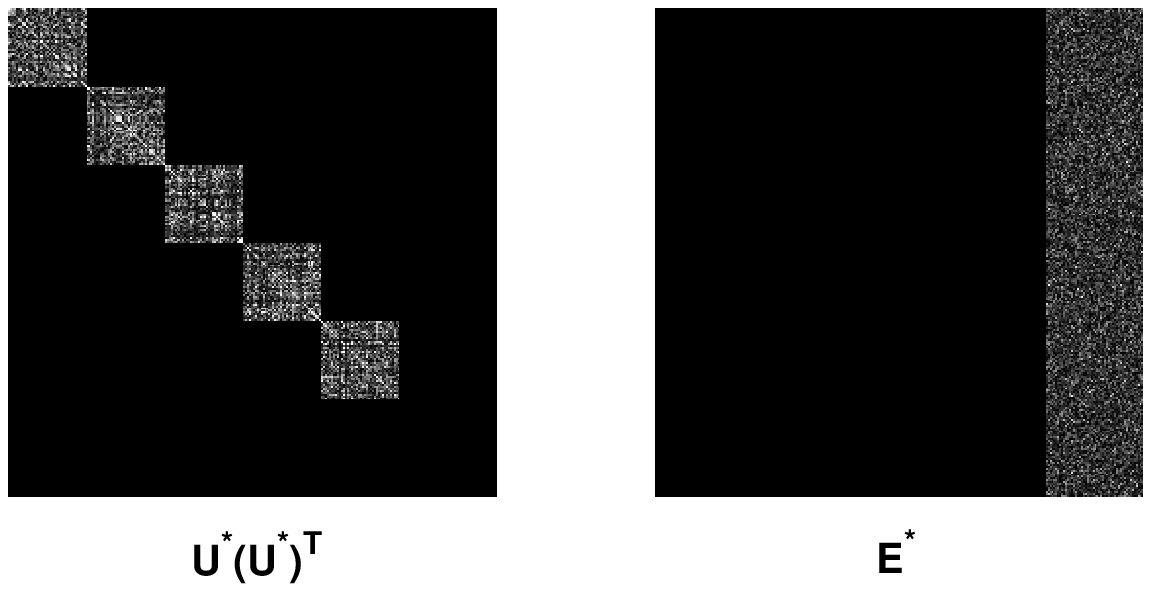}
\caption{\textbf{An example of the matrices $U^*(U^*)^T$ and $E^*$ computed from the data contaminated by outliers.} In a similar way as \cite{icml_2010_lrr}, we create 5 pairwise disjoint subspaces each of which is of dimension 4, and draw 40 samples (with ambient dimension 200) from each subspace. Then, 50 outliers are randomly generated from $\mathcal{N}(0,s)$, where the standard deviation $s$ is set to be three times as large as the averaged magnitudes of the samples. By choosing $0.16\leq\lambda\leq0.34$, LRR produces a solution $(Z^*,E^*)$ with the column space of $Z^*$ exactly recovering the row space of $X_0$, and the column supports of $E^*$ exactly identifying the indices of the outliers.}\label{fig:sim:outlier}\vspace{-0.25in}
\end{center}
\end{figure}
When an observed data sample is far away from the underlying subspaces, a typical regime is that this sample is from a different model other than subspaces, so called as an \emph{outlier} \footnote{Precisely, we define an outlier as a data vector that is independent to the samples drawn from the subspaces \cite{liu:2011:nips}.}. In this case, the data matrix $X$ contains two parts, one part consists of authentic samples (denoted by $X_0$) strictly drawn from the underlying subspaces, and the other part consists of outliers (denoted as $E_0$) that are not subspace members. To precisely describe this setting, we need to impose an additional constraint on $X_0$, that is,
\begin{eqnarray}\label{eq:settings:outlier}
\mathcal{P}_{\mathcal{I}_0}(X_0)=0,
\end{eqnarray}
where $\mathcal{I}_0$ is the indices of the outliers (i.e., the column supports of $E_0$). Furthermore, we use $n$ to denote the total number of data samples in $X$, $\gamma\triangleq{}|\mathcal{I}_0|/n$ the fraction of outliers, and $r_0$ the rank of $X_0$. With these notations, we have the following theorem which states that LRR can exactly recover the row space of $X_0$ and identify the indices of outliers as well.
\begin{theorem}[\cite{liu:2011:nips}]\label{theorem:lrr:outlier}
There exists $\gamma^*>0$ such that LRR with parameter $\lambda=3/(7\|X\|\sqrt{\gamma^*n})$ strictly succeeds, as long as $\gamma\leq\gamma^*$. Here, the success is in a sense that any minimizer $(Z^*,E^*)$ to \eqref{eq:lrrx:assu2} can produce
\begin{eqnarray}\label{eq:exact}
U^*(U^*)^T=V_0V_0^T&\textrm{and}&\mathcal{I}^*=\mathcal{I}_0,
\end{eqnarray}
where $U^*$ is the column space of $Z^*$, and $\mathcal{I}^*$ is column supports of $E^*$.
\end{theorem}

There are several importance notices in the above theorem. First, although the objective function \eqref{eq:lrrx:assu2} is not strongly convex and multiple minimizers may exist, it is proven that \emph{any} minimizer is effective for subspace clustering. Second, the coefficient matrix $Z^*$ itself does not recover $V_0V_0^T$ (notice that $Z^*$ is usually asymmetric except $E^*=0$), and it is the column space of $Z^*$ that recovers the row space of $X_0$. Third, the performance of LRR is measured by the value of $\gamma^*$ (the larger, the better), which depends on some data properties such as the incoherence and the extrinsic rank $r_0$ ($\gamma^*$ is larger when $r_0$ is lower). For more details, please refer to \cite{liu:2011:nips}.

Fig.\ref{fig:sim:outlier} shows some experimental results, which verify the conclusions of Theorem \ref{theorem:lrr:outlier}. Notice that the parameter setting $\lambda=3/(7\|X\|\sqrt{\gamma^*n})$ is based on the condition $\gamma\leq\gamma^*$ (i.e., the outlier fraction is smaller than a certain threshold), which is just a sufficient (but not necessary) condition for ensuring the success of LRR. So, in practice (even for synthetic examples) where $\gamma>\gamma^*$, it is possible that other values of $\lambda$ achieve better performances.
\subsubsection{Robustness to Sample-Specific Corruptions}
\begin{figure}
\begin{center}
\includegraphics[width=0.4\textwidth]{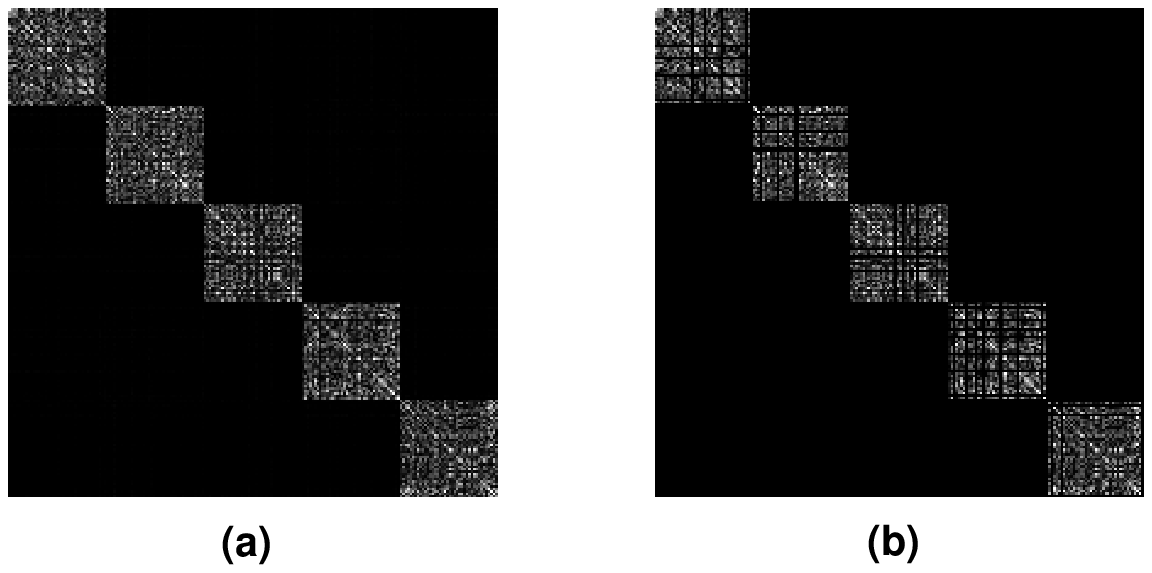}
\caption{\textbf{Two examples of the matrix $U^*(U^*)^T$ computed from the data corrupted by sample-specific corruptions.} (a) The magnitudes of the corruptions are set to be about 0.7 times as large as the samples. Considering $|U^*(U^*)^T|$ as an affinity matrix, the average affinity degree of the corrupted samples is about $40$, which means that the corrupted samples can be projected back onto their respective subspaces. (b) The magnitudes of the corruptions are set to be about 3.5 times as large as the samples. The affinity degrees of the corrupted samples are all zero, which means that the corrupted samples are treated as outliers. In these experiments, the data samples are generated in the same way as in Fig.\ref{fig:sim:outlier}. Then, 10\% samples are randomly chosen to be corrupted by additive errors of Gaussian distribution. For each experiment, the parameter $\lambda$ is carefully determined such that the column supports of $E^*$ identify the indices of the corrupted samples.}\label{fig:sim2}\vspace{-0.25in}
\end{center}
\end{figure}
For the phenomenon that an observed sample is away from the subspaces, another regime is that this sample is an authentic subspace member, but grossly corrupted. Usually, such corruptions only happen on a small fraction of data samples, so called as ``sample-specific'' corruptions. The modeling of sample-specific corruptions is the same as outliers, because in both cases $E_0$ has sparse column supports. So the formulation \eqref{eq:lrrx:assu2} is still applicable. However, the setting \eqref{eq:settings:outlier} is no longer valid, and thus LRR may not exactly recover the row space $V_0V_0^T$ in this case. Empirically, the conclusion of $\mathcal{I}^*=\mathcal{I}_0$ still holds \cite{icml_2010_lrr}, which means that the column supports of $E^*$ can identify the indices of the corrupted samples.

While both outliers and sample-specific corruptions \footnote{Unlike outlier, a corrupted sample is unnecessary to be independent to the clean samples.} are handled in the same way, a question is how to deal with the cases where the authentic samples are heavily corrupted to have similar properties as the outliers. If a sample is heavily corrupted so as to be independent from the underlying subspaces, it will be treated as an outlier in LRR, as illustrated in Fig.\ref{fig:sim2}. This is a reasonable manipulation. For example, it is appropriate to treat a face image as a non-face outlier if the image has been corrupted to be look like something else.
\vspace{-0.1in}\subsection{Robustness in the Presence of Noise, Outliers and Sample-Specific Corruptions}
\begin{figure}
\begin{center}
\includegraphics[width=0.4\textwidth]{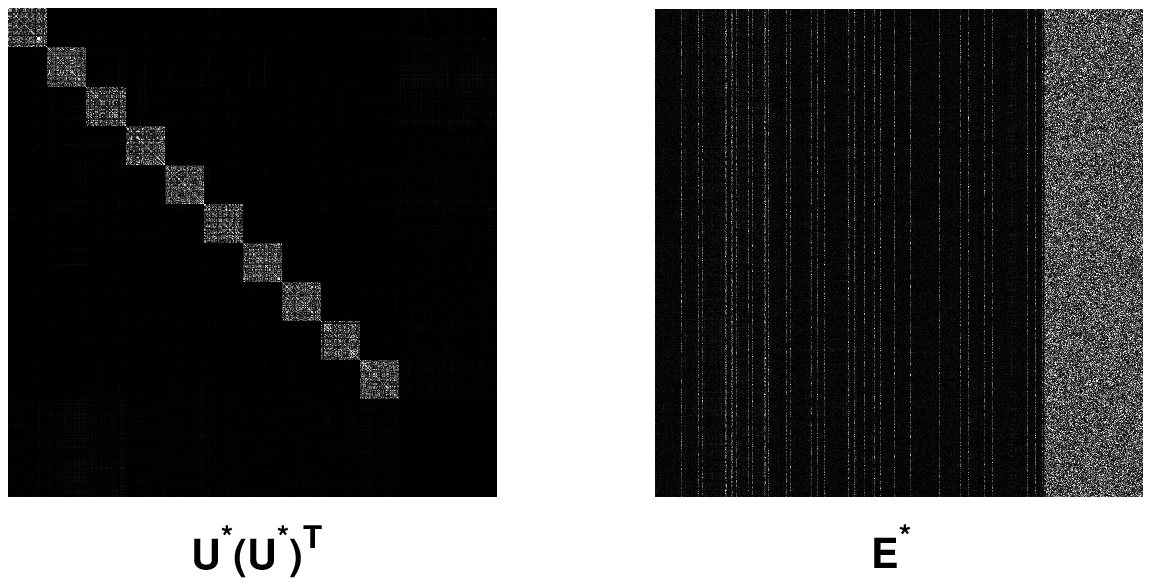}
\caption{\textbf{An example of the matrices $U^*(U^*)^T$ and $E^*$ computed from the data contaminated by noise, outliers and sample-specific corruptions.} In this experiment, first, we create 10 pairwise disjoint subspaces (each of which is of dimension 4) and draw 40 samples (with ambient dimension 2000) from each subspace. Second, we randomly choose 10\% samples to be grossly corrupted by large errors. The rest 90\% samples are slightly corrupted by small errors. Finally, as in Fig.\ref{fig:sim:outlier}, 100 outliers are randomly generated. The total amount of errors (including noise, sample-specific corruptions and outliers) is given by $\|E_0\|_F/\|X_0\|_F=0.63$. By setting $\lambda=0.3$, $U^*(U^*)^T$ approximately recovers $V_0V_0^T$ with error $\|U^*(U^*)^T-V_0V_0^T\|_F/\|V_0V_0^T\|_F=0.17$, and the column supports of $E^*$ accurately identify the indices of the outliers and corrupted samples. In contrast, the recover error produced by PCA is 0.66, and that by the RPCA method (using the best parameters) introduced in \cite{xu:2010:nips} is 0.23.}\label{fig:sim3}\vspace{-0.25in}
\end{center}
\end{figure}
When there is noise in the data, the column supports of $E_0$ are not strictly sparse. Nevertheless, the formulation \eqref{eq:lrrx:assu2} is still applicable, because the $\ell_{2,1}$ norm (which is relaxed from $\ell_{2,0}$ norm) can handle well the signals that approximately have sparse column supports. Since all observations may be contaminated, it is unlikely in theory that the row space $V_0V_0^T$ can be exactly recovered. So we target on near recovery in this case. By the triangle inequality of matrix norms, the following theorem can be simply proven without any assumptions.
\begin{theorem}\label{the:lrr:robustness:general}
Let the size of $X$ be $d\times{}n$, and the rank of $X_0$ be $r_0$. For any minimizer $(Z^*,E^*)$ to problem \eqref{eq:lrrx:assu2} with $\lambda>0$, we have
$$\|Z^*-V_0V_0^T\|_F\leq{}\min(d,n)+r_0.$$
\end{theorem}

Fig.\ref{fig:sim3} demonstrates the performance of LRR, in the presence of noise, outliers and sample-specific corruptions. It can be seen that the results produced by LRR are quite promising.

One may have noticed that the bound given in above theorem is somewhat loose. To obtain a more accurate bound in theory, one needs to relax the equality constraint of \eqref{eq:lrrx:assu2} into:
\begin{eqnarray*}
\min_{Z,E} ||Z||_*+\lambda{||E||_{2,1}}, \textrm{s.t. } \|X - XZ - E\|_F\leq{}\xi,
\end{eqnarray*}
where $\xi$ is a parameter for characterizing the amount of the dense noise (Fig.\ref{fig:noise}(a)) possibly existing in data. The above problem can be solved by ALM, in a similar procedure as Algorithm \ref{alg:alm}. However, the above formulation needs to invoke another parameter $\xi$, and thus we do not further explore it in this paper.
\vspace{-0.1in}\subsection{Algorithms for Subspace Segmentation, Model Estimation and Outlier Detection}
\subsubsection{Segmentation with Given Subspace Number}
\begin{algorithm}[tb]
   \caption{Subspace Segmentation}\label{alg:segmentation}
   \label{alg:lrr}
\begin{algorithmic}
   \STATE {\bfseries Input:} data matrix $X$, number $k$ of subspaces.
   \STATE \textbf{1.} obtain the minimizer $Z^*$ to problem \eqref{eq:lrrx:assu2}.
   \STATE \textbf{2.} compute the skinny SVD $Z^*=U^*\Sigma^*(V^*)^T$.
   \STATE \textbf{3.} construct an affinity matrix $W$ by \eqref{eq:w}.
   \STATE \textbf{4.} use $W$ to perform NCut and segment the data samples into $k$ clusters.
\end{algorithmic}
\end{algorithm}
After obtaining $(Z^*,E^*)$ by solving problem \eqref{eq:lrrx:assu2}, the matrix $U^*(U^*)^T$ that identifies the column space of $Z^*$ is useful for subspace segmentation. Let the skinny SVD of $Z^*$ as $U^*\Sigma^*(V^*)^T$, we define an affinity matrix $W$ as follows:
\begin{eqnarray}\label{eq:w}
[W]_{ij} = ([\tilde{U}\tilde{U}^T]_{ij})^2,
\end{eqnarray}
where $\tilde{U}$ is formed by $U^*(\Sigma^*)^{\frac{1}{2}}$ with normalized rows. Here, for obtaining better performance on corrupted data, we assign each column of $U^*$ a weight by multiplying $(\Sigma^*)^{\frac{1}{2}}$. Notice that when the data is clean, $\Sigma^*=\Id$ and thus this technique does not take any effects. The technical detail of using $(\cdot)^2$ is to ensure that the values of the affinity matrix $W$ are positive (note that the matrix $\tilde{U}\tilde{U}^T$ can have negative values). Finally, we could use the spectral clustering algorithms such as Normalized Cuts (NCut) \cite{tpami_2000_ncut} to segment the data samples into a given number $k$ of clusters. Algorithm \ref{alg:segmentation} summarizes the whole procedure of performing segmentation by LRR.
\subsubsection{Estimating the Subspace Number $k$}
 \begin{algorithm}[tb]
   \caption{Estimating the Subspace Number $k$}\label{alg:est_k}
   \label{alg:lrr}
\begin{algorithmic}
   \STATE {\bfseries Input:} data matrix $X$.
   \STATE \textbf{1.} compute the affinity matrix $W$ in the same way as in Algorithm \ref{alg:segmentation}.
   \STATE \textbf{2.} compute the Laplacian matrix $L=\Id-D^{-\frac{1}{2}}WD^{-\frac{1}{2}}$, where $D=\diag{\sum_{j}[W]_{1j},\cdots,\sum_{j}[W]_{nj}}$.
   \STATE \textbf{3.} estimate the subspace number by \eqref{eq:est_k}.
\end{algorithmic}
\end{algorithm}
Although it is generally challenging to estimate the number of subspaces (i.e., number of clusters), it is possible to resolve this model estimation problem due to the block-diagonal structure of the affinity matrix produced by specific algorithms \cite{cvpr_2009_ssc,Vidal:2004:cvpr,Huang:2004:cvpr}. While a strictly block-diagonal affinity matrix $W$ is obtained, the subspace number $k$ can be found by firstly computing the normalized Laplacian (denoted as $L$) matrix of $W$, and then counting the number of zero singular values of $L$. While the obtained affinity matrix is just near block-diagonal (this is the case in reality), one could predict the subspace number as the number of singular values smaller than a threshold. Here, we suggest a soft thresholding approach that outputs the estimated subspace number $\hat{k}$ by
\begin{eqnarray}\label{eq:est_k}
\hat{k}=n-\mathrm{int}(\sum_{i=1}^{n}f_{\tau}(\sigma_i)).
\end{eqnarray}
Here, $n$ is the total number of data samples, $\{\sigma_i\}_{i=1}^{n}$ are the singular values of the Laplacian matrix $L$, $\mathrm{int}(\cdot)$ is the function that outputs the nearest integer of a real number, and $f_{\tau}(\cdot)$ is a soft thresholding operator defined as
\begin{eqnarray*}
f_{\tau}(\sigma)=\left\{
\begin{array}{ll} 1, & \mbox{if $\sigma\geq\tau$},\\
\log_2(1+\frac{\sigma^2}{\tau^2}), & \mbox{otherwise,}
\end{array}\right.
\end{eqnarray*}
where $0<\tau<1$ is a parameter. Algorithm \ref{alg:est_k} summarizes the whole procedure of estimating the subspace number based on LRR.
\subsubsection{Outlier Detection}
As shown in Theorem \ref{theorem:lrr:outlier}, the minimizer $E^*$ (with respect to the variable $E$) can be used to detect the outliers that possibly exist in data. This can be simply done by finding the nonzero columns of $E^*$, when all or a fraction of data samples are clean (i.e., Assumption 1 and Assumption 2). For the cases where the learnt $E^*$ only approximately has sparse column supports, one could use thresholding strategy; that is, the $i$-th data vector of $X$ is judged to be outlier if and only if
\begin{eqnarray}\label{eq:outlier_detection}
\|[E^*]_{:,i}\|_2>\delta,
\end{eqnarray}
where $\delta>0$ is a parameter.

Since the affinity degrees of the outliers are zero or close to being zero (see Fig.\ref{fig:sim:outlier} and Fig.\ref{fig:sim3}), the possible outliers can be also removed by discarding the data samples whose affinity degrees are smaller than a certain threshold. Such a strategy is commonly used in spectral-type methods \cite{cvpr_2009_ssc,Lerman:2011:Lp}. Generally, the underlying principle of this strategy is essential the same as \eqref{eq:outlier_detection}. Comparing to the strategy of characterizing the outliers by affinity degrees, there is an advantage of using $E^*$ to indicate outliers; that is, the formulation \eqref{eq:lrrx:assu2} can be easily extended to include more priors, e.g., the multiple visual features as done in \cite{chen:2011:iccv,lang:2011:tip}.
\vspace{-0.15in}\section{Experiments}\label{sec:exp}
LRR has been used to achieve state-of-the-art performance in several applications such as motion segmentation \cite{liu:2011:iccv}, image segmentation \cite{chen:2011:iccv}, face recognition \cite{liu:2011:iccv} and saliency detection \cite{lang:2011:tip}. In the experiments of this paper, we shall focus on analyzing the essential aspects of LRR, under the context of subspace segmentation and outlier detection.
\vspace{-0.1in}\subsection{Experimental Data}
\begin{table}[t]
\caption{Some information about Hopkins155.}\label{tb:hop155:inf}
\begin{center}
\begin{tabular}{|c|c|c|c|c|}
\hline
    & data      & \# of data    & \# of     & error\\
    & dimension & samples       & subspaces  & level\\
\hline
max & 201       & 556           & 3         & 0.0130\\
min & 31        & 39            & 2         & 0.0002\\
mean& 59.7      & 295.7         & 2.3       & 0.0009\\
std. & 20.2      & 140.8         & 0.5       & 0.0012\\
\hline
\end{tabular}
\end{center}
\end{table}
\subsubsection{Hopkins155}
To verify the segmentation performance of LRR, we adopt for experiments the Hopkins155 \cite{hopkin155} motion database, which provides an extensive benchmark for testing various subspace segmentation algorithms. In Hopkins155, there are 156 video sequences along with the features extracted and tracked in all the frames. Each sequence is a sole dataset (i.e., data matrix) and so there are in total 156 datasets of different properties, including the number of subspaces, the data dimension and the number of data samples. Although the outliers in the data have been manually removed and the overall error level is low, some sequences (about 10 sequences) are grossly corrupted and have notable error levels. Table \ref{tb:hop155:inf} summarizes some information about Hopkins155. For a sequence represented as a data matrix $X$, its error level is estimated by its rank-$r$ approximation:  $\|X-U_r\Sigma_rV_r^T\|_F/\|X\|_F$, where $\Sigma_r$ contains the largest $r$ singular values of $X$, and $U_r$ (resp. $V_r$) is formed by taking the top $r$ left (resp. right) singular vectors. Here, we set $r=4k$ ($k$ is the subspace number of the sequence), due to the fact that the rank of each subspace in motion data is at most 4.
\begin{figure}
\begin{center}
\centerline{\includegraphics[width=0.4\textwidth]{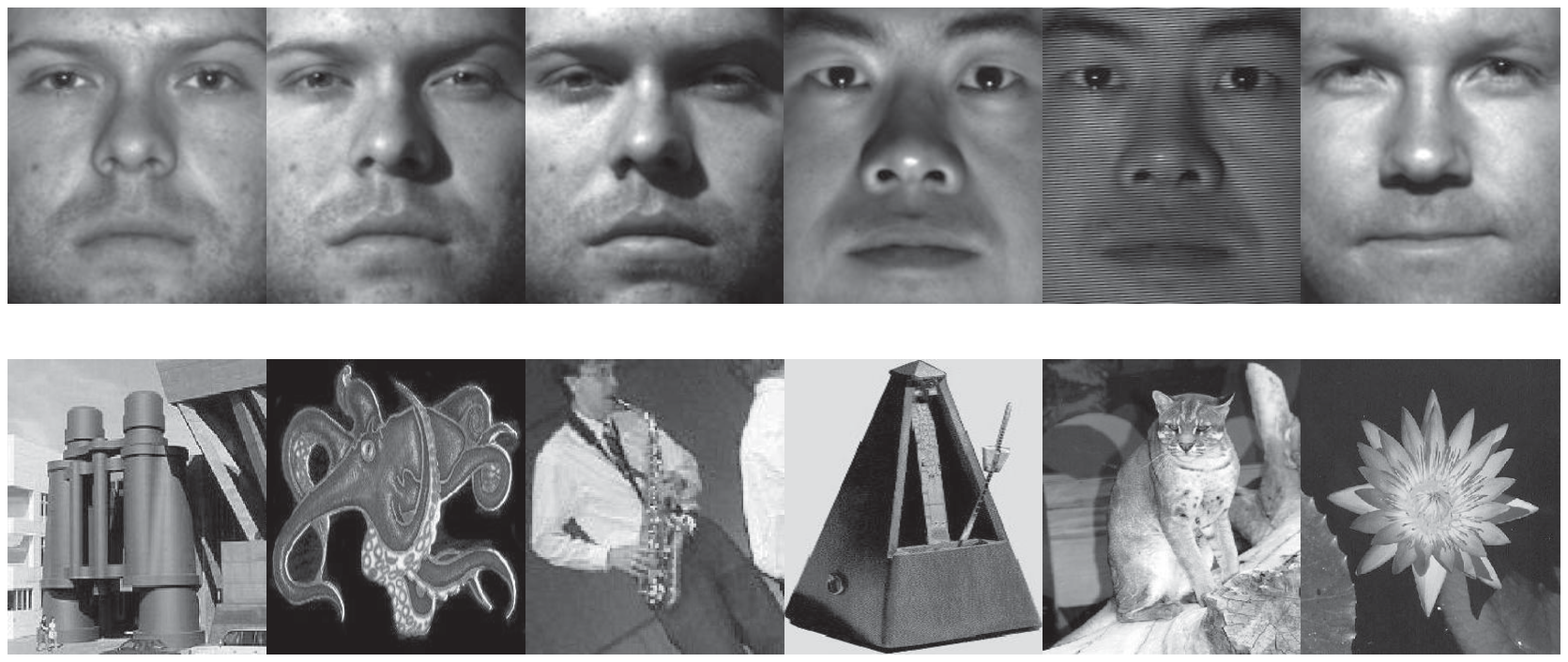}} \caption{Examples of the images in the Yale-Caltech dataset.} \vspace{-0.25in}\label{fig:yaleb-cal:example}
\end{center}
\end{figure}
\subsubsection{Yale-Caltech} To test LRR's effectiveness in the presence of outliers and corruptions, we create a dataset by combining Extended Yale Database B \cite{eyaleb} and Caltech101 \cite{Fei-Fei:2004:LGV}. For Extended Yale Database B, we remove the images pictured under extreme light conditions. Namely, we only use the images with view directions smaller than 45 degrees and light
source directions smaller than 60 degrees, resulting in 1204 authentic samples approximately drawn from a union of 38 low-rank subspaces (each face class corresponds to a subspace). For Caltech101, we only select the classes containing no more than 40 images, resulting in 609 non-face outliers. Fig.\ref{fig:yaleb-cal:example} shows some examples of this dataset.
\vspace{-0.1in}\subsection{Baselines and Evaluation Metrics}\label{sec:baselines}Due to the close connections between PCA and LRR, we choose PCA and RPCA methods as the baselines. Moreover, some previous subspace segmentation methods are also considered.
\subsubsection{PCA (i.e., SIM)} The PCA method is widely used for dimension reduction. Actually, it can also be applied to subspace segmentation and outlier detection as follows: first, we use SVD to obtain the rank-$r$ ($r$ is a parameter) approximation of the data matrix $X$, denoted as $X\approx{}U_r\Sigma_rV_r^T$; second, we utilize $V_rV_r^T$, which is an estimation of the true SIM $V_0V_0^T$, for subspace segmentation in a similar way as Algorithm \ref{alg:segmentation} (the only difference is the estimation of SIM); finally, we compute $E_r=X-U_r\Sigma_rV_r^T$ and use $E_r$ to detect outliers according to \eqref{eq:outlier_detection}.
\subsubsection{RPCA} As an improvement over PCA, the robust PCA (RPCA) methods can also do subspace segmentation and outlier detection. In this work, we consider two RPCA methods introduced in \cite{journal_2009_rpca2} and \cite{xu:2010:nips}, which are based on minimizing
\begin{eqnarray*}
\min_{D,E} \norm{D}_*+\lambda\|E\|_{\ell},\text{ s.t. }X=D+E.
\end{eqnarray*}
In \cite{journal_2009_rpca2}, the $\ell_1$ norm is used to characterize random corruptions, so referred to as ``RPCA$_1$''. In \cite{xu:2010:nips}, the $\ell_{2,1}$ norm is adopted for detecting outliers, so referred to as ``RPCA$_{2,1}$''. The detailed procedures for subspace segmentation and outlier detection are almost the same as the PCA case above. The only difference is that $V_r$ is formed from the skinny SVD of $D^*$ (not $X$), which is obtained by solving the above optimization problem. Note here that the value of $r$ is determined by the parameter $\lambda$, and thus one only needs to select $\lambda$.
\subsubsection{SR} LRR has similar appearance as SR, which has been applied to subspace segmentation \cite{cvpr_2009_ssc}. For fair comparison, in this work we implement an $\ell_{2,1}$-norm based SR method that computes an affinity matrix by minimizing
\begin{eqnarray*}
\min_{Z,E} \norm{Z}_1+\lambda\|E\|_{2,1},\text{ s.t. }X=XZ+E,[Z]_{ii}=0.
\end{eqnarray*}
Here, SR needs to enforce $[Z]_{ii}=0$ to avoid the trivial solution $Z=\Id$. After obtaining a minimizer $(Z^*,E^*)$, we use $W=|Z^*|+|(Z^*)^T|$ as the affinity matrix to do subspace segmentation. The procedure of using
$E^*$ to perform outlier detection is the same as LRR.
\subsubsection{Some other Methods}
We also consider for comparison some previous subspace segmentation methods, including Random Sample Consensus (RANSAC) \cite{cacm_1981_ransac}, Generalized PCA (GPCA) \cite{siam_2008_gpca}, Local Subspace Analysis (LSA) \cite{eccv_2006_lsa}, Agglomerative Lossy Compression (ALC) \cite{motion_pami_2010_Rene}, Sparse Subspace Clustering (SSC) \cite{cvpr_2009_ssc}, Spectral Clustering (SC) \cite{sc:iccv:2009}, Spectral Curvature Clustering (SCC) \cite{Chen:2009:SCC}, Multi Stage Learning (MSL) \cite{sugaya:2004:ieice}, Locally Linear Manifold Clustering (LLMC) \cite{vidal:2007:cvpr}, Local Best-fit Flats (LBF) \cite{Zhang:2011:lbf} and Spectral LBF (SLBF) \cite{Zhang:2011:lbf}.
\subsubsection{Evaluation Metrics}\label{sec:metrix}
Segmentation accuracy (error) is used to measure the performance of segmentation. The areas under the receiver operator characteristic (ROC) curve, known as AUC, is used for for evaluating the quality of outlier detection. For more details about these two evaluation metrics, please refer to Appendix.
\vspace{-0.1in}\subsection{Results on Hopkins155}
\begin{figure}
\begin{center}
\includegraphics[width=0.45\textwidth]{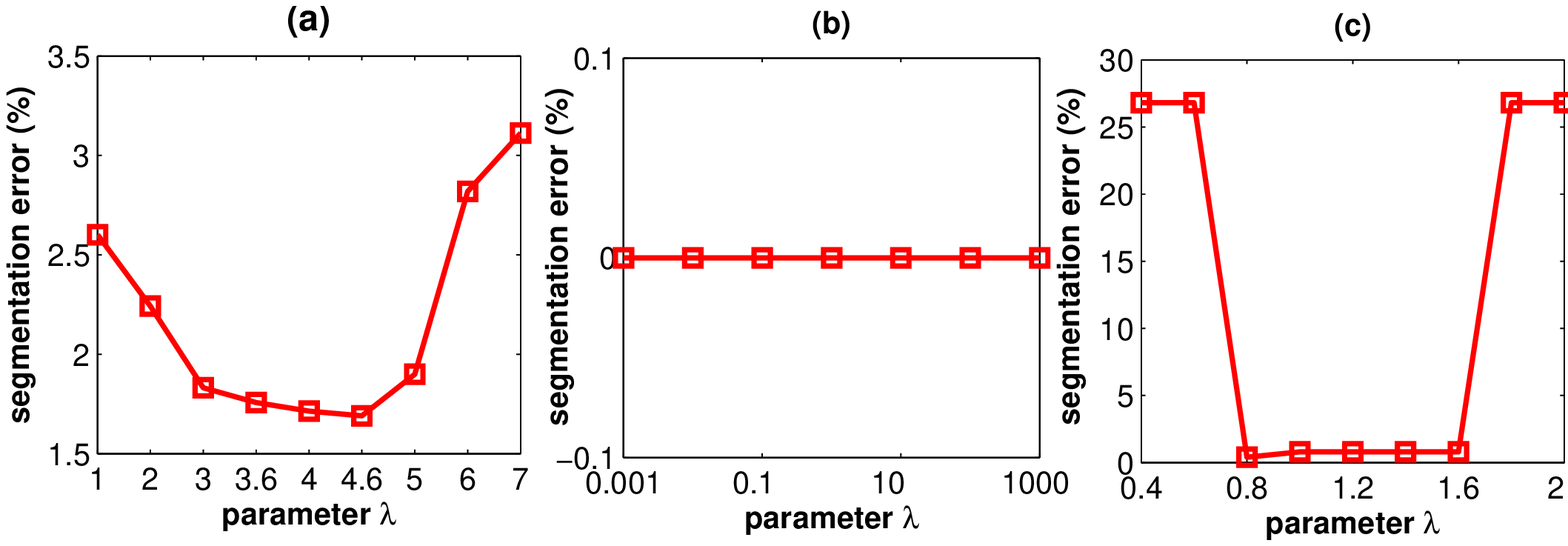}
\caption{\textbf{The influences of the parameter $\lambda$ of LRR.} (a) On all 156 sequences of Hopkins155, the overall segmentation performance is equally good while $3\leq\lambda\leq5$. (b) On the $43$-th sequence, the segmentation error is always 0 for $0.001\leq\lambda\leq1000$. (c) On the $62$-th sequence, the segmentation performance is good only when $0.8\leq\lambda\leq1.6$.}\label{fig:para:hop}\vspace{-0.25in}
\end{center}
\end{figure}
\subsubsection{Choosing the Parameter $\lambda$}\label{sec:lambda}
The parameter $\lambda>0$ is used to balance the effects of the two parts in problem \eqref{eq:lrrx:assu2}. In general, the choice of this parameter depends on the prior knowledge of the error level of data. When the errors are slight, we should use relatively large $\lambda$; when the errors are heavy, we should set $\lambda$ to be relatively small.

Fig.\ref{fig:para:hop}(a) shows the evaluation results over all 156 sequences in Hopkins155: while $\lambda$ ranges from $1$ to $6$, the segmentation error only varies from $1.69\%$ to $2.81\%$; while $\lambda$ ranges from $3$ to $5$, the segmentation error almost remains unchanged, slightly varying from 1.69\% to 1.87\%. This phenomenon is mainly due to two reasons as follows. First, on most sequences (about 80\%) which are almost clean and easy to segment, LRR could work well by choosing $\lambda$ arbitrarily, as exemplified in Fig.\ref{fig:para:hop}(b). Second, there is an ``invariance'' in LRR, namely Theorem \ref{the:lrr:solution} implies that the minimizer to problem \eqref{eq:lrrx:assu2} always satisfies $Z^*\in\spn{X^T}$. This implies that the solution of LRR can be partially stable while $\lambda$ is varying.

The analysis above does not deny the importance of model selection. As shown in Fig.\ref{fig:para:hop}(c), the parameter $\lambda$ can largely affect the segmentation performance on some sequences. Actually, if we turn $\lambda$ to the best for each sequence, the overall error rate is only 0.07\%. Although this number is achieved in an ``impractical'' way, it verifies the significance of selecting the parameter $\lambda$, especially when the data is corrupted. For the experiments below, we choose $\lambda=4$ for LRR.
\subsubsection{Segmentation Performance}
\begin{table}[t]
\caption{Segmentation results (on Hopkins155) of PCA, RPCA$_1$, RPCA$_{2,1}$, SR and LRR.}\label{tb:hop155:pcas}
\begin{center}
\begin{tabular}{|lccccl|}\hline
\multicolumn{6}{|c|}{\textbf{segmentation errors (\%) over all 156 sequences}}\\
      & PCA   & RPCA$_1$ & RPCA$_{2,1}$ & SR & LRR\\
\hline
mean  & 4.56  & 4.13     &3.26          &3.89 & \textbf{1.71}\\
std.  & 10.80 & 10.37    &9.09          &7.70 & \textbf{4.85} \\
max   & 49.78 & 45.83    &47.15         &\textbf{32.57} & 33.33\\\hline
\multicolumn{6}{|c|}{\textbf{average run time (seconds) per sequence}}\\
&\textbf{0.2} &0.8       &0.8           &4.2  &1.9\\\hline
\end{tabular}\vspace{-0.15in}
\end{center}
\end{table}

In this subsection, we show LRR's performance in subspace segmentation with the subspace number given. For comparison, we also list the results of PCA, RPCA$_1$, RPCA$_{2,1}$ and SR (these methods are introduced in Section \ref{sec:baselines}). Table \ref{tb:hop155:pcas} illustrates that LRR performs better than PCA and RPCA. Here, the advantages of LRR are mainly due to its methodology. More precisely, LRR \emph{directly} targets on recovering the row space $V_0V_0^T$, which provably determines the segmentation results. In contrast, PCA and RPCA methods are designed for recovering the column space $U_0U_0^T$, which is designed for dimension reduction. One may have noticed that RPCA$_{2,1}$ outperforms PCA and RPCA$_1$. If we use instead the $\ell_1$ norm to regularize $E$ in \eqref{eq:lrrx:assu2}, the segmentation error is 2.03\% ($\lambda=0.6$, optimally determined). These illustrate that the errors in this database tend to be sample-specific.

Besides the superiorities in segmentation accuracy, another advantage of LRR is that it can work well under a wide range of parameter settings, as shown in Fig.\ref{fig:para:hop}. Whereas, RPCA methods are sensitive to the parameter $\lambda$. Taking RPCA$_{2,1}$ for example, it achieves an error rate of 3.26\% by choosing $\lambda=0.32$. However, the error rate increases to 4.5\% at $\lambda=0.34$, and 3.7\% at $\lambda=0.3$.

The efficiency (in terms of running time) of LRR is comparable to PCA and RPCA methods. Theoretically, the computational complexity (with regard to $d$ and $n$) of LRR is the same as RPCA methods. LRR costs more computational time because its optimization procedure needs more iterations than RPCA to converge.
\subsubsection{Performance of Estimating Subspace Number}\begin{table}[t]
\caption{Results (on Hopkins155) of estimating the subspace number.}\label{tb:hop155:estk}
\begin{center}
\begin{tabular}{cccc}\hline
\# total & \# predicted & prediction rate (\%) &absolute error\\
156      & 121          &77.6      &0.25\\\hline
\end{tabular}
\begin{tabular}{ccccccc}\hline
\multicolumn{7}{c}{\textbf{influences of the parameter $\tau$}}\\
parameter $\tau$          & 0.06 & 0.07 & \textbf{0.08} & 0.09& 0.10 & 0.11\\
prediction rate           & 66.7 & 71.2 & \textbf{77.6} & 75.0& 72.4 &71.2\\
absolute error            &0.37  &0.30  &\textbf{0.25}  &0.26 & 0.29 &0.30\\\hline
\end{tabular}\vspace{-0.15in}
\end{center}
\end{table}
Since there are 156 sequences in total, this database also provides a good benchmark for evaluating the effectiveness of Algorithm \ref{alg:est_k}, which is to estimate the number of subspaces underlying a collection of data samples. Table \ref{tb:hop155:estk} shows the results. By choosing $\tau=0.08$, LRR correctly predicts the true subspace number of 121 sequences. The absolute error (i.e., $|\hat{k}-k|$) averaged over all sequences is $0.25$. These results illustrate that it is hopeful to resolve the problem of estimating the subspace number, which is a challenging model estimation problem.
\subsubsection{Comparing to State-of-the-art Methods}
\begin{table}[t]
\caption{Segmentation errors (\%) on Hopkins155 (155 sequences).}
\label{tb:hop155}
\begin{center}
\begin{tabular}{|cccccc|}\hline
      & GPCA &RANSAC &MSL & LSA  &LLMC \\
mean  &10.34 &9.76   &5.06& 4.94 &4.80 \\\hline
      &PCA   &LBF    &ALC & SCC  &SLBF\\
mean  &4.47  &3.72   &3.37& 2.70 &1.35\\\hline
      &      &       & \multicolumn{3}{|c|}{LRR}\\\cline{4-6}
      &SSC   & SC    & \multicolumn{1}{|c}{\cite{vidal:2011:cvpr}} & \cite{liu:2011:iccv} & this paper\\\hline
mean  &1.24  & 1.20  & 1.22 & \textbf{0.85} & 1.59\\
\hline
\end{tabular}\vspace{-0.15in}
\end{center}
\end{table}
Notice that previous methods only report the results for 155 sequences. After discarding the degenerate sequence, the error rate of LRR is 1.59\% which is comparable to the state-of-the-art methods, as shown in Table \ref{tb:hop155}. The performance of LRR can be further improved by refining the formulation \eqref{eq:lrrx:assu2}, which uses the observed data matrix $X$ itself as the dictionary. When the data is corrupted by dense noise (this is usually true in reality), this certainly is not the best choice. In \cite{vidal:2011:cvpr} and \cite{rsi:tsp}, a non-convex formulation is adopted to learn the original data $X_0$ and its row space $V_0V_0^T$ simultaneously:
\begin{eqnarray*}
\min_{D,Z,E}\|Z\|_*+\lambda\|E\|_1\textrm{ s.t. }X=D+E,D=DZ,
\end{eqnarray*}
where the unknown variable $D$ is used as the dictionary. This method can achieve an error rate of 1.22\%. In \cite{liu:2011:iccv}, it is explained that the issues of choosing dictionary can be relieved by considering the unobserved, hidden data. Furthermore, it is deduced that the effects of hidden data can be approximately modeled by the following convex formulation:
\begin{eqnarray*}
\min_{Z,L,E}\|Z\|_*+\|L\|_*+\lambda\|E\|_1\textrm{ s.t. }X=XZ+LX+E,
\end{eqnarray*}
which intuitively integrates subspace segmentation and feature extraction into a unified framework. This method can achieve an error rate of 0.85\%, which outperforms other subspace segmentation algorithms.

While several methods have achieved an error rate below 3\% on Hopkins155, subspace segmentation problem is till far from solved. A long term difficult is how to solve the model selection problems, e.g., estimating the parameter $\lambda$ of LRR. Also, it would not be trivial to handle more complicated datasets that contain more noise, outliers and corruptions.
\vspace{-0.1in}\subsection{Results on Yale-Caltech}
The goal of this test is to identify 609 non-face outliers and segment the rest 1204 face images into 38 clusters. The performance of segmentation and outlier detection is evaluated by segmentation accuracy (ACC) and AUC, respectively. While investigating segmentation performance, the affinity matrix is computed from all images, including both the face images and non-face outliers. However, for the convenience of evaluation, the outliers and the corresponding affinities are removed (according to the ground truth) before using NCut to obtain the segmentation results.
\begin{table}[t]
\caption{Segmentation accuracy (ACC) and AUC comparison on the Yale-Caltech dataset.}
\label{tb:yale}
\begin{center}
\begin{tabular}{|cccccc|}\hline
                &PCA            &RPCA$_1$       &RPCA$_{2,1}$       &SR            &LRR\\\hline
ACC (\%)        &77.15          &82.97          &83.72              &73.17         &\textbf{86.13}\\
AUC             &0.9653         &0.9819         &0.9863             &0.9239        &\textbf{0.9927}\\
time (sec.)     &\textbf{0.6}   &60.8           &59.2               &383.5         &152.6\\
\hline
\end{tabular}\vspace{-0.15in}
\end{center}
\end{table}

We resize all images into $20\times20$ pixels and form a data matrix of size $400\times1813$. Table \ref{tb:yale} shows the results of PCA, RPCA, SR and LRR. It can be seen that LRR is better than PCA and RPCA methods, in terms of both subspace segmentation and outlier detection. These experimental results are consistent with Theorem \ref{theorem:lrr:outlier}, which shows that LRR has a stronger guarantee than RPCA methods in performance. Notice that SR is  behind the others \footnote{The results (for outlier detection) in Table \ref{tb:yale} are obtained by using the strategy of \eqref{eq:outlier_detection}. While using the strategy of checking the affinity degree, the results produced by SR is even worse, only achieving an AUC of 0.81 by using the best parameters.}. This is because the presence or absence of outliers is unnecessary to notably alert the sparsity of the reconstruction coefficients, and thus it is hard for SR to handle well the data contaminated by outliers.
\begin{figure}
\begin{center}
\includegraphics[width=0.3\textwidth]{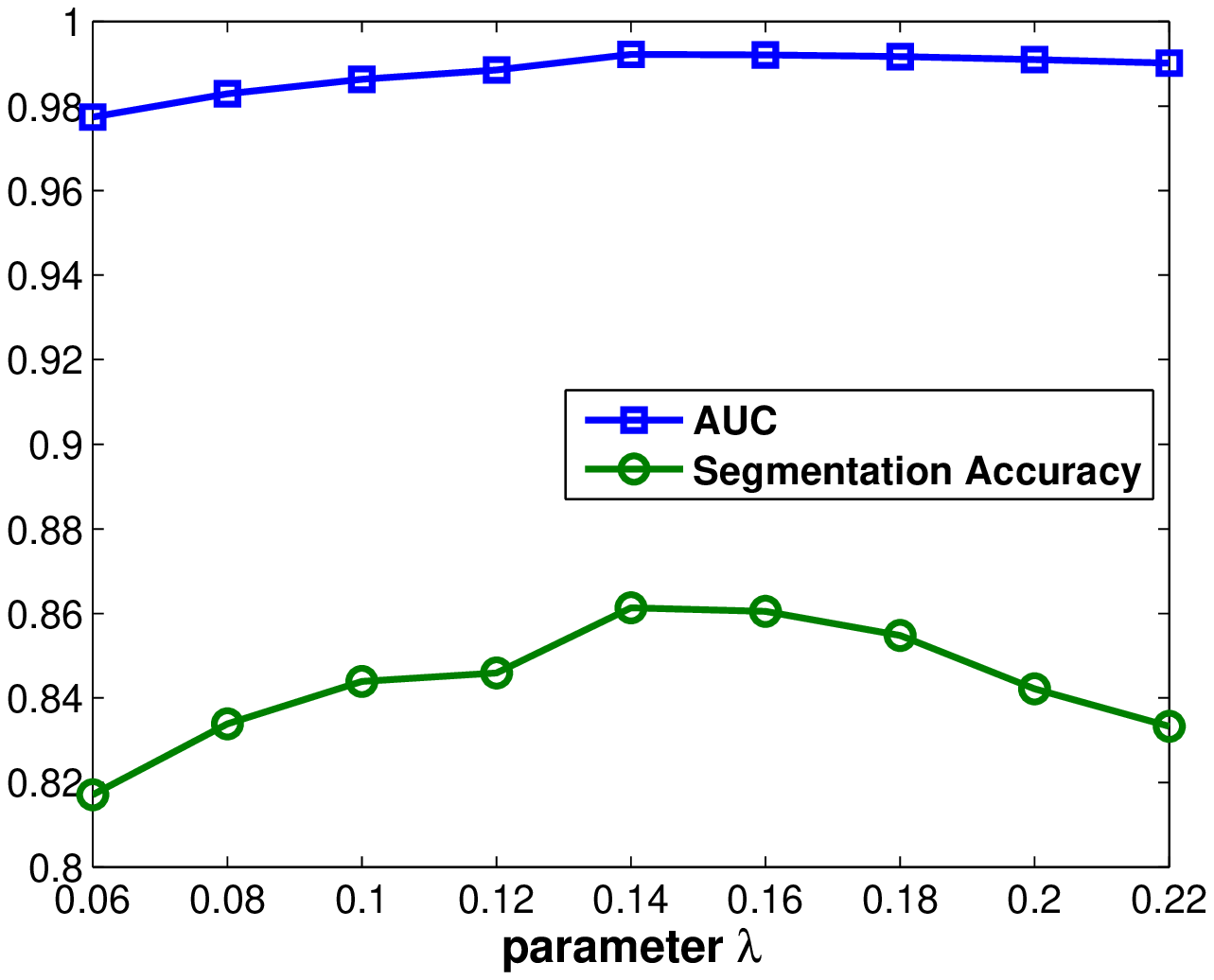}
\caption{\textbf{The influences of the parameter $\lambda$ of LRR.} These results are collected from the Yale-Caltech dataset. All images are resized to $20\times20$ pixels.}\label{fig:para:yale}\vspace{-0.25in}
\end{center}
\end{figure}

Fig.\ref{fig:para:yale} shows the performance of LRR while the parameter $\lambda$ varies from 0.06 to 0.22. Notice that LRR is more sensitive to $\lambda$ on this dataset than on Hopkins155. This is because the error level of Hopkins155 is quite low (see Table \ref{tb:hop155:inf}), whereas, the Yale-Caltach dataset contains outliers and corrupted images (see Fig.\ref{fig:yaleb-cal:example}).
\begin{figure}
\begin{center}
\includegraphics[width=0.4\textwidth]{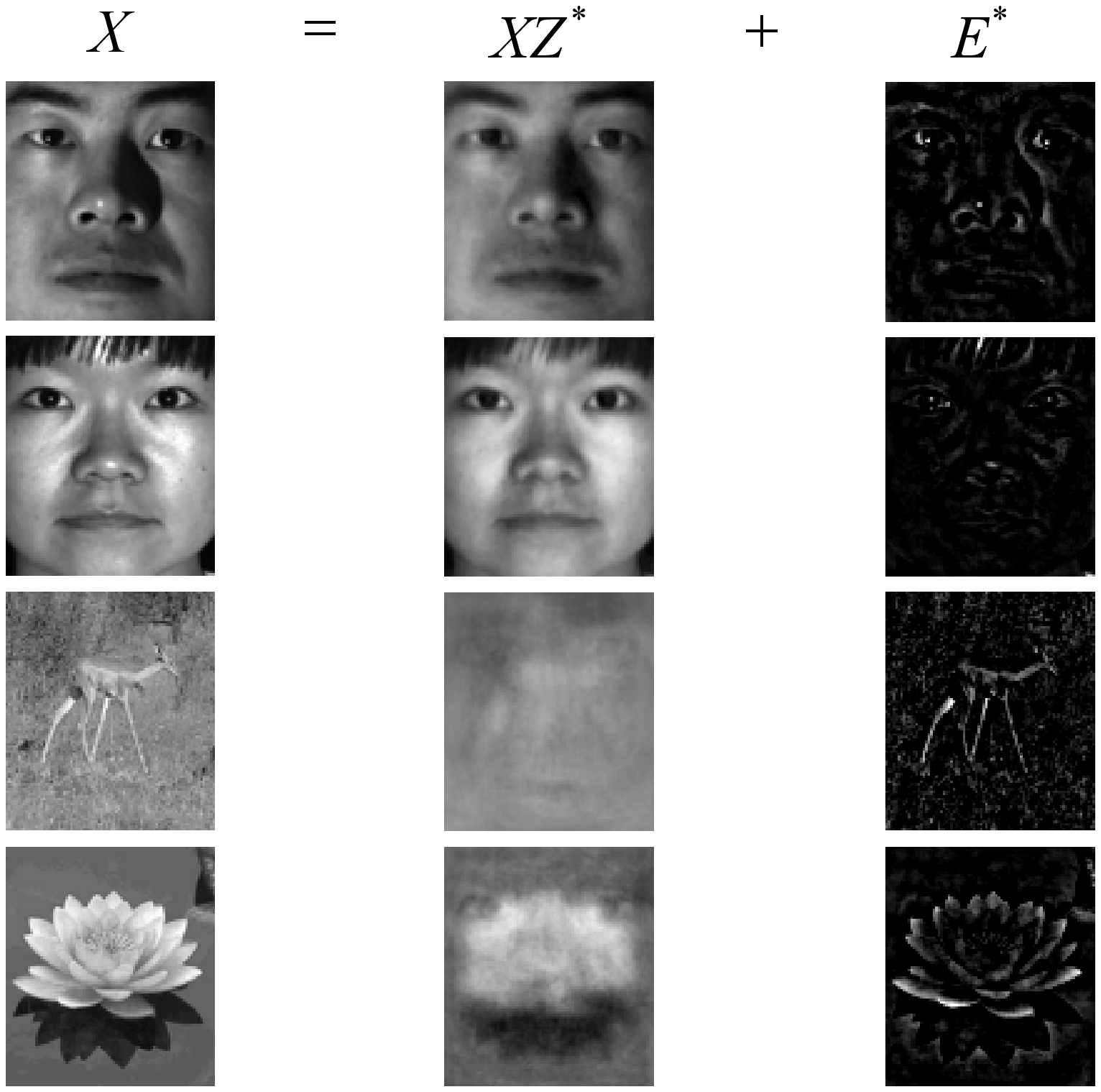}
\caption{\textbf{Some examples of using LRR to correct the errors in the Yale-Caltech dataset.} Left: the original data matrix $X$; Middle: the corrected data $XZ^*$; Right: the error $E^*$.}\label{fig:yale:res}\vspace{-0.25in}
\end{center}
\end{figure}

To visualize LRR's effectiveness in error correction, we create another data matrix with size $8064\times1813$ by resizing all images into $96\times84$. Fig.\ref{fig:yale:res} shows some results produced by LRR. It is worth noting that the ``error'' term $E^*$ can contain ``useful'' information, e.g., the eyes and salient objects. Here, the principle is to decompose the data matrix into a low-rank part and a sparse part, with the low-rank part ($XZ^*$) corresponding to the principal features of the whole dataset, and the sparse part ($E^*$) corresponding to the rare features which cannot be modeled by low-rank subspaces. This implies that it is possible to use LRR to extract the discriminative features and salient regions, as done in face recognition \cite{liu:2011:iccv} and saliency detection \cite{lang:2011:tip}.
\vspace{-0.15in}\section{Conclusion and Future Work}\label{sec:conclusion}
In this paper we proposed low-rank representation (LRR) to identify the subspace structures from corrupted data. Namely, our goal is to segment the samples into their respective subspaces and correct the possible errors simultaneously. LRR is a generalization of the recently established RPCA methods \cite{journal_2009_rpca2,xu:2010:nips}, extending the recovery of corrupted data from single subspace to multiple subspaces. Also, LRR generalizes the approach of Shape Interaction Matrix (SIM), giving a way to define an SIM between two different matrices (see Theorem \ref{theorem:unique:nonoise}), and providing a mechanism to recover the true SIM (or row space) from corrupted data. Both theoretical and experimental results show the effectiveness of LRR. However, there still remain several problems for future work:
\begin{itemize}
\item[$\bullet$] It may achieve significant improvements by learning a dictionary $A$, which partially determines the solution of LRR. In order to exactly recover the row space $V_0$, Theorem \ref{the:lrr:solution} illustrates that the dictionary $A$ must satisfy the condition of $V_0\in\spn{A^T}$. When the data is only contaminated by outliers, this condition can be obeyed by simply choosing $A=X$. However, this choice cannot ensure the validity of $V_0\in\spn{A^T}$ while the data contains other types of errors, e.g., dense noise.

\item[$\bullet$] The proofs of Theorem \ref{theorem:lrr:outlier} are specific to the case of $A=X$. As a future direction, it is interesting to see whether the technique presented can be extended to general dictionary matrices other than $X$.

\item[$\bullet$] A critical issue in LRR is how to estimate or select the parameter $\lambda$. For the data contaminated by various errors such as noise, outliers and corruptions, the estimation of $\lambda$ is quite challenging.

\item[$\bullet$] The subspace segmentation should not be the only application of LRR. Actually, it has been successfully used in the applications other than segmentation, e.g., saliency detection \cite{lang:2011:tip}. In general, the presented LRR method can be extended to solve various applications well.
\end{itemize}
\section*{Appendix}\label{sec:appendix}
\subsection{Terminologies}\label{subsec:ter}
In this subsection, we introduce some terminologies used in the paper.
\subsubsection{Block-Diagonal Matrix} In this paper, a matrix $M$ is called block-diagonal if it has the form as in (1). For the matrix $M$ which itself is not block-diagonal but can be transformed to be block-diagonal by simply permuting its rows and/or columns, we also say that $M$ is block-diagonal. In summary, we say that a matrix $M$ is block-diagonal whenever there exist two permutation matrices $P_1$ and $P_2$ such that $P_1MP_2$ is block-diagonal.
\subsubsection{Union and Sum of Subspaces} For a collection of $k$ subspaces $\{\mathcal{S}_1,\mathcal{S}_2,\cdots,\mathcal{S}_k\}$, their union is defined by $\cup_{i=1}^{k}\mathcal{S}_i=\{y:y\in\mathcal{S}_j, \textrm{ for some } 1\leq{}j\leq{}k\}$, and their sum is defined by $\sum_{i=1}^{k}\mathcal{S}_i=\{y:y=\sum_{j=1}^{k}y_j,y_j\in\mathcal{S}_j\}$. If any $y\in\sum_{i=1}^{k}\mathcal{S}_i$ can be uniquely expressed as $y=\sum_{j=1}^{k}y_j$, $y_j\in\mathcal{S}_j$, then the sum is also called the directed sum, denoted as $\sum_{i=1}^{k}\mathcal{S}_i=\oplus_{i=1}^{k}\mathcal{S}_i$.
\subsubsection{Independent Subspaces} A collection of $k$ subspaces $\{\mathcal{S}_1,\mathcal{S}_2,\cdots,\mathcal{S}_k\}$ are independent if and only if $\mathcal{S}_i\cap\sum_{j\neq{i}}\mathcal{S}_j=\{0\}$ (or $\sum_{i=1}^{k}\mathcal{S}_i=\oplus_{i=1}^{k}\mathcal{S}_i$). When the subspaces are of low-rank and the ambient dimension is high, the independent assumption is roughly equal to the pairwise disjoint assumption; that is $\mathcal{S}_i\cap\mathcal{S}_j=\{0\},\forall{i\neq{j}}$.
\subsubsection{Full SVD and Skinny SVD} For an $m\times{n}$ matrix $M$ (without loss of generality, assuming $m\leq{n}$), its Singular Value Decomposition (SVD) is defined by $M=U[\Sigma,0]V^T$,
where $U$ and $V$ are orthogonal matrices and $\Sigma=\diag{\sigma_1,\sigma_2,\cdots,\sigma_m}$ with $\{\sigma_i\}_{i=1}^m$ being singular values.
The SVD defined in this way is also called the \emph{full} SVD. If we only keep the positive singular values, the reduced form is called the \emph{skinny} SVD. For a matrix $M$ of rank $r$, its skinny SVD is computed by $M=U_r\Sigma_rV_r^T$, where $\Sigma_r=\diag{\sigma_1,\sigma_2,\cdots,\sigma_r}$ with $\{\sigma_i\}_{i=1}^r$ being positive singular values. More precisely, $U_r$ and $V_r$ are formed by taking the first $r$ columns of $U$ and $V$, respectively.
\subsubsection{Pseudoinverse} For a matrix $M$ with skinny SVD $U\Sigma{}V^T$, its pseudoinverse is uniquely defined by
$$M^\dag=V\Sigma^{-1}U^T.$$

\subsubsection{Column Space and Row Space} For a matrix $M$, its column (resp. row) space is the linear space spanned by its column (resp. row) vectors. Let the skinny SVD of $M$ be $U\Sigma{}V^T$, then $U$ (resp. $V$) are orthonormal bases of the column (resp. row) space, and the corresponding orthogonal projection is given by $UU^T$ (resp. $VV^T$). Since $UU^T$ (resp. $VV^T$) is \emph{uniquely} determined by the column (resp. row) space, sometimes we also use $UU^T$ (resp. $VV^T$) to refer to the column (resp. row) space.
\subsubsection{Affinity Degree} Let $M$ be a symmetric affinity matrix for a collection of $n$ data samples, the affinity degree of the $i$-th sample is defined by $\#\{(j):[M]_{ij}\neq0\}$, i.e., the number of samples connected to the $i$-th sample.

\subsection{Proofs}
\subsubsection{Proof of Theorem 4.1}
The proof of Theorem 4.1 is based on the following three lemmas.
\begin{lemma}\label{lemma:basic}
Let $U$, $V$ and $M$ be matrices of compatible dimensions. Suppose both $U$ and $V$ have orthogonal columns,
i.e., $U^TU=\Id$ and $V^TV=\Id$, then we have
\begin{eqnarray*}
\norm{M}_* = \|UMV^T\|_*.
\end{eqnarray*}
\end{lemma}
\begin{proof}Let the full SVD of $M$ be $M=U_M\Sigma_MV_M^T$, then $UMV^T=(UU_M)\Sigma_{M}(VV_M)^T$.
As $(UU_M)^T(UU_M)=\Id$ and $(VV_M)^T(VV_M)=\Id$, $(UU_M)\Sigma_{M}(V_MV)^T$ is actually an SVD of $UMV^T$. By the definition of the nuclear norm, we have $\norm{M}_* = \trace{\Sigma_M}=\norm{UMV^T}_*$.
\end{proof}

\begin{lemma}\label{lemma:abcd:0condition}
For any four matrices $B$, $C$, $D$ and $F$ of compatible dimensions, we have
\begin{eqnarray*}
\norm{\left[\begin{array}{cc}
B&C\\
D&F\\
\end{array}\right]}_* \geq \norm{B}_*,
\end{eqnarray*}
where the equality holds if and only if $C=0,D=0$ and $F=0.$
\end{lemma}
\begin{proof}
The proof is simply based on the following fact: for any two matrices $M_1$ and $M_2$, we have
\begin{eqnarray*}
\|[M_1,M_2]\|_*\geq{}\|M_1\|_*&\textrm{and/or}&\|[M_1;M_2]\|_*\geq{}\|M_1\|_*
\end{eqnarray*}
and the equality can hold if and only if $M_2=0$.
\end{proof}

\begin{lemma}\label{lemma:unique:nonoise}
Let $U$, $V$ and $M$ be given matrices of compatible dimensions. Suppose both $U$ and $V$ have orthogonal columns,
i.e., $U^TU=\Id$ and $V^TV=\Id$, then the following optimization problem
\begin{eqnarray}\label{eq:lemma:problem}
\min_{Z} \norm{Z}_*, & \textrm{s.t.} & U^TZV=M,
\end{eqnarray}
has a unique minimizer $Z^* = UMV^T$.
\end{lemma}
\begin{proof}
First, we prove that $\norm{M}_*$ is the minimum objective function value and $Z^*=UMV^T$ is a minimizer. For any
feasible solution $Z$, let $Z = U_Z\Sigma_ZV_Z^T$ be its full SVD. Let $B=U^TU_Z$ and $C=V_Z^TV$. Then the constraint
$U^TZV=M$ is equal to
\begin{equation}\label{eq:lemma:BSzC=M}
B\Sigma_ZC = M.
\end{equation}
Since $BB^T=\Id$ and $C^TC=\Id$, we can find the orthogonal complements \footnote{When $B$ and/or $C$ are already orthogonal
matrices, i.e., $B_{\bot}=\emptyset$ and/or $C_{\bot}=\emptyset$, our proof is still valid.} $B_{\bot}$ and $C_{\bot}$ such that
\begin{eqnarray*}
\left[\begin{array}{c}
B\\
B_{\bot}\\
\end{array}\right] &\textrm{and}&[C,C_{\bot}]
\end{eqnarray*}
are orthogonal matrices. According to the unitary invariance of the nuclear norm, Lemma \ref{lemma:abcd:0condition} and
\eqref{eq:lemma:BSzC=M}, we have
\begin{eqnarray*}
\norm{Z}_* &=& \norm{\Sigma_Z}_* = \norm{\left[\begin{array}{c}
B\\
B_{\bot}\\
\end{array}\right]\Sigma_Z[C,C_{\bot}]}_*\\
&=&\norm{\left[\begin{array}{cc}
B\Sigma_ZC&B\Sigma_ZC_{\bot}\\
B_{\bot}\Sigma_ZC&B_{\bot}\Sigma_ZC_{\bot}\\
\end{array}\right]}_*\\
&\geq&\norm{B\Sigma_ZC}_*=\norm{M}_*,
\end{eqnarray*}
Hence, $\norm{M}_*$ is the minimum objective function value of problem \eqref{eq:lemma:problem}. At the same time, Lemma \ref{lemma:basic} proves that
$\norm{Z^*}_*=\norm{UMV^T}_*=\norm{M}_*$. So $Z^*=UMV^T$ is a minimizer to problem \eqref{eq:lemma:problem}.

Second, we prove that $Z^*=UMV^T$ is the unique minimizer. Assume that $Z_1=UMV^T+H$ is another optimal solution. By $U^TZ_1V=M$, we
have
\begin{equation}\label{eq:lemma:UHV=0}
U^THV = 0.
\end{equation}
Since $U^TU=\Id$ and $V^TV=\Id$, similar to above, we can construct two orthogonal matrices: $[U,U_{\bot}]$ and $[V,V_{\bot}]$.
By the optimality of $Z_1$, we have
\begin{eqnarray*}
\norm{M}_*&=&\norm{Z_1}_*=\norm{UMV^T+H}_*\\
&=&\norm{\left[\begin{array}{c}
U^T\\
U_{\bot}^T\\
\end{array}\right](UMV^T+H)[V,V_{\bot}]}_*\\
&=&\norm{\left[\begin{array}{cc}
M&U^THV_{\bot}\\
U_{\bot}^THV&U_{\bot}^THV_{\bot}\\
\end{array}\right]}_*\\
&\geq&\norm{M}_*.
\end{eqnarray*}
According to Lemma \ref{lemma:abcd:0condition}, the above equality can hold if and only if
\begin{equation*}
U^THV_{\bot}=U_{\bot}^THV=U_{\bot}^THV_{\bot}=0.
\end{equation*}
Together with \eqref{eq:lemma:UHV=0}, we conclude that $H=0$. So the optimal solution is unique.
\end{proof}

It is worth noting that Lemma \ref{lemma:unique:nonoise} allows us to get closed-form solutions to a class of nuclear norm minimization problems, and leads to a simple proof of Theorem 4.1.

\begin{proof}{\bf (of Theorem 4.1)}
Since $X\in{\spn{A}}$, we have $\rank{[X,A]}=\rank{A}$. Let's define $V_X$ and $V_A$ as follows: Compute the skinny SVD of the horizontal concatenation of $X$ and $A$, denoted as $[X,A]=U\Sigma{}V^T$, and partition $V$ as $V=[V_X;V_A]$ such that $X=U\Sigma{}V_X^T$ and $A=U\Sigma{V_A}^T$ (note that $V_A$ and $V_X$ may be not column-orthogonal). By this definition, it can be concluded that the matrix $V_A^T$ has full row rank. That is, if the skinny SVD of $V_A^T$ is $U_1\Sigma_1V_1^T$, then $U_1$ is an orthogonal matrix. Through some simple computations, we
have
\begin{equation}\label{eq:theorem:inv_va}
V_A(V_A^TV_A)^{-1} = V_1\Sigma_1^{-1}U_1^T.
\end{equation}
Also, it can be calculated that the constraint $X=AZ$ is equal to $V_X^T=V_A^TZ$, which is also equal to
$\Sigma_1^{-1}U_1^TV_X^T=V_1^TZ$. So problem (5) is equal to the following optimization problem:
\begin{eqnarray*}
\min_{Z} \norm{Z}_*, & \textrm{s.t.} & V_1^TZ=\Sigma_1^{-1}U_1^TV_X^T.
\end{eqnarray*}
By Lemma \ref{lemma:unique:nonoise} and \eqref{eq:theorem:inv_va}, problem (5) has a unique minimizer
\begin{eqnarray*}
Z^* &=& V_1\Sigma_1^{-1}U_1^TV_X^T= V_A(V_A^TV_A)^{-1}V_X^T.
\end{eqnarray*}

Next, it will be shown that the above closed-form solution can be further simplified. Notice that $V_A^T = \Sigma^{-1}U^TA$ and $V_X^T = \Sigma^{-1}U^TX$. Then we have
\begin{eqnarray*}
Z^* &=& A^TU\Sigma^{-1}(\Sigma^{-1}U^TAA^TU\Sigma^{-1})^{-1}\Sigma^{-1}U^TX\\
    &=& A^TU(U^TAA^TU)^{-1}U^TX\\
    &=&(U^TA)^{\dag}U^TX\\
    &=&A^{\dag}X,
\end{eqnarray*}
where the last equality is due to that $(U^TA)^{\dag}U^T=(\Sigma_AV_A^T)^{\dag}U^T=(U\Sigma_AV_A^T)^{\dag}=A^{\dag}$.
\end{proof}

\subsubsection{Proof of Corollary 4.1}
\begin{proof}By $X\in\spn{A}$, we have $\rank{A^{\dag}X}=\rank{X}$. Hence, $\rank{Z^*}=\rank{X}$. At the same time, for any feasible
solution $Z$ to problem (5), we have $\rank{Z}\geq\rank{AZ}=\rank{X}$. So, $Z^*$ is also optimal to
problem (4).
\end{proof}

\subsubsection{Proof of Theorem 4.2}
The proof of Theorem 4.2 is based on the following well-known lemma.
\begin{lemma}\label{lemma:abcd}
For any four matrices $B$, $C$, $D$ and $F$ of compatible dimensions, we have
\begin{eqnarray*}
\norm{\left[\begin{array}{cc}
B&C\\
D&F\\
\end{array}\right]}_* \geq \norm{\left[\begin{array}{cc}
B&0\\
0&F\\
\end{array}\right]}_* =\norm{B}_* + \norm{F}_*.
\end{eqnarray*}
\end{lemma}

The above lemma allows us to lower-bound the objective value at any solution $Z$ by the value of the
block-diagonal restriction of $Z$, and thus leads to a simple proof of Theorem 4.2.

\begin{proof}
Let $Z^*$ be the optimizer to problem (5). Form a block-diagonal matrix $W$ by
setting
\begin{eqnarray*}
[W]_{ij} = \left\{ \begin{array}{ll} [Z]_{ij}, & [A]_{:,i}\text{ and }[X]_{:,j} \; \text{belong to}\\
& \textrm{the same subspace,}\\ 0, & \text{otherwise.} \end{array} \right.
\end{eqnarray*}
Write $Q=Z^*-W$. For any data vector $[X]_{:,j}$, without loss of generality, suppose $[X]_{:,j}$ belongs to the $i$-th
subspace; i.e., $[AZ^*]_{:,j} \in S_i$. Then by construction, we have $[AW]_{:,j} \in S_i$ and $[AQ]_{:,j} \in
\oplus_{m \ne i} S_m$. But $[A Q]_{:,j} = [X]_{:,j} - [A W]_{:,j} \in S_i$. By independence, we have $S_i \cap
\oplus_{m \ne i} S_m = \{ 0 \}$, and so $[A Q]_{:,j} = 0, \;\forall \, j$.

Hence, $A Q = 0$, and $W$ is feasible for (5). By Lemma \ref{lemma:abcd}, we have $\|Z^* \|_* \ge \| W \|_*$. Also, by the uniqueness of the minimizer (see Theorem 4.1), we
conclude that $Z^*=W$ and hence $Z^*$ is block-diagonal.

Again, by the uniqueness of the minimizer $Z^*$, we can conclude that for all $i$'s, $Z_i^*$ is also the unique
minimizer to the following optimization problem:
\begin{eqnarray*}
\min_{J} \norm{J}_*, &\textrm{s.t.} & X_i=A_iJ.
\end{eqnarray*}
By Corollary 4.1, we conclude that $\rank{Z_i^*} = \rank{X_i}$.
\end{proof}
\subsubsection{Proof of Theorem 4.3}
\begin{proof}
Note that the LRR problem (7) always has feasible solution(s), e.g., $(Z=0,E=X)$ is feasible. So, an optimal solution, denoted as $(Z^*,E^*)$, exists. By Theorem 4.1, we have
\begin{eqnarray*}
Z^*&=&\arg\min_{Z}\|Z\|_* \textrm{\hspace{0.1in}s.t.\hspace{0.1in}}X-E^*=AZ\\
&=&A^{\dag}(X-E^*),
\end{eqnarray*}
which simply leads to $Z^*\in\spn{A^T}$.
\end{proof}
\subsubsection{Proof of Theorem 5.3}
\begin{proof} Let the skinny SVD of $X$ be $U\Sigma{}V^T$. It is simple to see that $(VV^T,0)$ is feasible to problem (9). By the convexity of (9), we have
\begin{eqnarray*}
\|Z^*\|_*\leq\|Z^*\|_*+\lambda\|E^*\|&\leq&\|VV^T\|_*\\
&=&\rank{X}\\
&\leq&\min(d,n).
\end{eqnarray*}
Hence,
\begin{eqnarray*}
\|Z^*-V_0V_0^T\|_F&\leq&\|Z^*-V_0V_0^T\|_*\leq\|Z^*\|_*+\|V_0V_0^T\|_*\\
&=&\|Z^*\|_*+r_0\leq\min(d,n)+r_0.
\end{eqnarray*}
\end{proof}
\subsection{Evaluation Metrics}
\subsubsection{Segmentation Accuracy (or Error)} The segmentation results can be evaluated in a similar way as classification results. Nevertheless, since segmentation methods cannot provide the class label for each cluster, a postprocessing step is needed to assign each cluster a label. A commonly used strategy is to try every possible label vectors that satisfy the segmentation results. The final label vector is chosen as the one that best matches the ground truth classification results. Such a \emph{global search} strategy is precise, but inefficient when the subspace number $k$ is large. Namely, the computational complexity is $k!$, which is higher than $2^k$ for $k\geq2$. Hence, we suggest a \emph{local search} strategy as follows: given the ground truth classification results, the label of a cluster is the index of the ground truth class that contributes the maximum number of samples to the cluster. This local search strategy is quite efficient because its computational complexity is only $O(k)$, and can usually produce the same evaluation results as global search. Nevertheless, it is possible that two different clusters are assigned with the same label. So, we use the local search strategy only when $k\geq10$.

\subsubsection{Receiver Operator Characteristic} To evaluate the effectiveness of outlier detection without choosing a parameter $\delta$ for (14), we consider the receiver operator characteristic (ROC), which is widely used to evaluate the performance of binary classifiers. The ROC curve is obtained by trying all possible thresholding values, and for each value, plotting the true positives
rate on the Y-axis against the false positive rate value on the X-axis. The areas under the ROC curve, known as AUC, provides a number for evaluating the quality of outlier detection. Note that the AUC score is the larger the better, and always ranges between 0 and 1.
\section*{Acknowledges}
We would like to acknowledge to support of "NExT Research Center" funded by MDA, Singapore, under the research grant: WBS:R-252-300-001-490.
\bibliographystyle{IEEEtran}
\bibliography{tpami_lrr_simple}
\end{document}